\documentclass[letter,11pt]{article}
\usepackage[margin=1in]{geometry}
\usepackage[utf8]{inputenc}

\usepackage{times}
\usepackage[T1]{fontenc}    
\usepackage{hyperref}       
\usepackage{url}            
\usepackage{booktabs}       
\usepackage{amsfonts}       
\usepackage{nicefrac}       
\usepackage{microtype}      
\usepackage{graphicx}
\usepackage{amsmath}
\usepackage{amssymb}
\usepackage{amsthm}
\usepackage{csquotes}
\usepackage{thm-restate}
\usepackage{cleveref}
\usepackage{paralist}
\usepackage{xspace}
\usepackage[vlined,linesnumbered,ruled]{algorithm2e}
 
\numberwithin{equation}{section}
\usepackage{enumerate}
\usepackage{listings}
\usepackage{tikz}
\usepackage{csquotes}
\usepackage{cite}
\usepackage{thmtools,thm-restate} 
\usepackage{paralist} 
\usepackage{pgfplots}
\pgfplotsset{compat=newest}
\usepackage{thm-restate}
\usepackage{subcaption}
\usepackage{graphics,graphicx,xcolor,color,url}    

\usepackage{tikz}
\usepackage{pgfplots}
\usetikzlibrary{external,shapes,patterns,arrows,graphs,decorations.pathreplacing,calc,backgrounds}
\tikzset{vertex/.style={circle, draw, fill,inner sep=0pt, minimum width=4pt},
	cross/.style={cross out, draw=black, minimum size=2*(#1-\pgflinewidth), inner sep=0pt, outer sep=0pt},
	cross/.default={3pt},}

\DeclareMathOperator{\EX}{\mathbb{E}}
\DeclareMathOperator{\PR}{\mathbb{P}}

\newcommand{\ALG}{\mathrm{ALG}}
\newcommand{\OPT}{\mathrm{OPT}}
\newcommand{\opt}{\OPT}
\newcommand{\alg}{\ALG}

\newcommand{\minset}{{\sc MinSet}\xspace}
\newcommand{\maxset}{{\sc MaxSet}\xspace}
\newcommand{\fs}{\ensuremath{\mathcal{S}}}
\newcommand{\ui}{\ensuremath{\mathcal{I}}}

\newcommand{\mincover}{{\sc MinCover}\xspace}
\newcommand{\mms}{\ensuremath{\mathcal{M}}}
\newcommand{\mmu}{\ensuremath{\mathcal{U}}}
\newcommand{\setcover}{{\sc SetCover}\xspace}

\newcommand{\gs}{\ensuremath{g_s}} 
\newcommand{\gc}{\ensuremath{g_c}} 

\newcommand{\ogs}{\ensuremath{\bar{g}_s}} 
\newcommand{\ogc}{\ensuremath{\bar{g}_c}} 

\newcommand{\alphagreedy}{\textsc{($\alpha,\beta,\gamma$)-Greedy}\xspace} 

\newcommand{\grset}{\ensuremath{\rho}} 
\newcommand{\grsetu}{\ensuremath{\bar{\rho}}} 

\newcommand{\lw}{0.5mm}

\newcommand{\intervalcolor}[5]{
	\draw[#5] (#2, #4) node[#5,anchor=east]{#1} -- (#3, #4);
	\draw[#5] (#2, #4-0.1) -- (#2, #4+0.1);
	\draw[#5] (#3, #4-0.1) -- (#3, #4+0.1);
}
\newcommand{\interval}[4]{
	\intervalcolor{#1}{#2}{#3}{#4}{black}
}

\newcommand{\drawreal}[2]{
	\draw[black!40!olive] (#1, #2) circle (3pt);
}

\newcommand{\intervalr}[5]{
	\interval{#1}{#2}{#3}{#4}
	\drawreal{#5}{#4}
}

\newcommand{\vertDOTS}{\begin{array}{lll}
		\bullet\\
		\bullet\\
		\bullet\\
	\end{array}
}

\theoremstyle{theorem}
\newtheorem{theorem}{Theorem}[section]
\newtheorem{lemma}[theorem]{Lemma}
\newtheorem{coro}[theorem]{Corollary}

\newtheorem{definition}[theorem]{Definition}

\usepackage[textsize=scriptsize,textwidth=2.1cm,disable]{todonotes}

\newcommand{\nnew}[1]{#1}
\newcommand{\jnew}[1]{#1}

\begin{document}

\title{Set Selection under Explorable Stochastic Uncertainty via Covering Techniques}
\author{Nicole Megow\inst{1}\orcidID{0000-0002-3531-7644} \and
	Jens Schl\"oter\inst{1}\orcidID{0000-0003-0555-4806}}

\author{%
	Nicole Megow\thanks{University of Bremen, Faculty of Mathematics and Computer Science, \texttt{\{nicole.megow,jschloet\}@uni-bremen.de}} \and Jens Schl{\"o}ter\footnotemark[1]
}

\maketitle

\begin{abstract}
Given subsets of uncertain values, we study the problem of identifying the subset of minimum total value (sum of the uncertain values) by querying as few values as possible.
This \emph{set selection problem} falls into the field of \emph{explorable uncertainty} and is of intrinsic importance therein 
as it implies strong adversarial lower bounds for a wide range of interesting combinatorial problems such as knapsack and matchings.
We consider a stochastic problem variant and give algorithms that, in expectation, improve upon these adversarial lower bounds.
The key to our results is to prove a strong structural connection to a seemingly unrelated covering problem with uncertainty in the constraints via a linear programming formulation.
We exploit this connection to derive an algorithmic framework that can be used to solve both problems under uncertainty, obtaining nearly tight bounds on the competitive ratio. This is the first non-trivial stochastic result concerning the sum of unknown values without further structure known for the set. With our novel methods, we lay the foundations for solving more general problems in the area of explorable uncertainty.
\end{abstract}

\section{Introduction}

\noindent In the setting of  \emph{explorable uncertainty}, we consider optimization problems with uncertainty in the numeric input parameters. 
Instead of having access to the  numeric values, we are given \emph{uncertainty intervals} that contain the precise values. 
Each uncertainty interval can be queried, which reveals the corresponding precise value. 
The goal is to adaptively query intervals until we have sufficient information to optimally (or approximately) solve the underlying optimization problem, while minimizing the number of queries.

This paper mainly considers the \emph{set selection problem} (\minset) under explorable uncertainty.
In this problem, we are given a set of $n$ uncertain values represented by uncertainty intervals $\ui = \{I_1,\ldots,I_n\}$ and a family of $m$ sets $\fs = \{S_1,\ldots,S_m\}$ with $S \subseteq \ui$ for all $S \in \fs$. 
A value $w_i$ lies in its uncertainty interval $I_i$, is initially unknown, and can be revealed via a query. 
The value of a subset $S$ is $w({S}) = \sum_{I_i \in S} w_i$. Our goal is to determine a subset of minimum value as well as the corresponding value by using a minimal number of queries.
It can be seen as an optimization problem with uncertainty in the coefficients of the objective function: 
\begin{equation}
	\label{eq:minset}
	\tag{{\sc SetSelIP}} 
	\begin{array}{llll}
		\min &\sum_{j=1}^m x_j &\sum_{I_i \in S_j} w_i\\
		\text{s.t. }& \sum_{j=1}^m x_j &= 1 &\\
		& \hfill x_j & \in \{0,1\}& \forall j \in \{1,\ldots,m\}.
	\end{array}
\end{equation}
Since the precise $w_i$'s are uncertain, we do not always have sufficient information to just compute an optimal solution to~\eqref{eq:minset}  and instead might have to execute queries in order to determine such a solution. An algorithm for \minset under uncertainty can adaptively query intervals until it has sufficient information to determine an optimal solution to~\eqref{eq:minset}. Adaptivity in this context means that the algorithm can take previous query results into account to decide upon the next query.

In this paper, we consider the stochastic problem variant, where we assume that all values $w_i$ are drawn independently at random from their intervals $I_i$ according to unknown distributions $d_i$. 
Since there are instances that cannot be solved without querying the entire input, we analyze an algorithm $\alg$ in terms of its {\em competitive ratio}: for the set of problem instances $\mathcal{J}$, it is defined as $\max_{J \in \mathcal{J}} {\EX[\alg(J)]}/{\EX[\opt(J)]}$, where $\alg(J)$ is the number of queries needed by $\alg$ to solve instance $J$, and $\opt(J)$ is the minimum number of queries necessary to solve the instance.

\minset is a fundamental problem and of intrinsic importance within the field of explorable uncertainty. 
The majority of existing works
considers the \emph{adversarial setting}, where query outcomes are not stochastic but returned in a worst-case manner.
Selection type problems have been studied in the adversarial setting and constant (matching) upper and lower bounds are known, e.g., for selecting the minimum~\cite{kahan91queries}, 
the $k$-th smallest element~\cite{kahan91queries,feder03medianqueries}, 
a minimum spanning tree~\cite{erlebach08steiner_uncertainty,erlebach14mstverification,megow17mst,Erlebach22Learning}, sorting~\cite{halldorsson19sortingqueries} and geometric problems~\cite{bruce05uncertainty}. However,  these problems essentially boil down to comparing {\em single} uncertainty intervals and identifying the minimum of two unknown values.
Once we have to compare two (even disjoint) sets 
and the corresponding {\em sums} of unknown values, 
no deterministic algorithm can have a better adversarial competitive ratio than~$n$, the number of uncertainty intervals. This has been shown by Erlebach et al.~\cite{erlebach16cheapestset} for \minset, and it
implies adversarial lower bounds for classical combinatorial problems, 
such as, knapsack~\cite{meissner18querythesis} and~matchings~\cite{meissner18querythesis}, and solving ILPs with uncertainty in the cost coefficients~\cite{meissner18querythesis} as in~\eqref{eq:minset} above.
Thus, solving \minset under stochastic uncertainty is an important step towards obtaining improved results for this range of problems. As a main result, we provide substantially better algorithms for \minset under stochastic uncertainty. This is a key step for breaching adversarial lower bounds for a wide range of problems.

For the {\em stochastic setting}, the only related results we are aware of concern sorting~\cite{ChaplickHLT21} and
the problem of finding the minimum in each set of a given collection of sets~\cite{BampisDEdLMS21}.
Asking for the {\em sum} of unknown values is substantially different.

\subsection{The Covering Point of View}

Our key observation is that we can view \minset as a covering problem with uncertainty in the constraints.
To see this, we focus on the structure of the uncertainty intervals and how a query affects it. 
We assume that each interval $I_i \in \ui$ is either open (\emph{non-trivial}) or \emph{trivial}, i.e., $I_i=(L_i,U_i)$ or $I_i=\{w_i\}$; 
a standard technical assumption in explorable uncertainty.
In the latter case, $L_i = U_i = w_i$.
We call $L_i$ and $U_i$ \emph{lower and upper limit}.
If $S$ contains only trivial uncertainty intervals, then we define $I_S = [L_S,U_S] = \{w(S)\}$ and call $I_S$ \emph{trivial}.
Otherwise, we define $I_S = (L_S, U_S)$ .
Clearly, the value $w(S)$ of a set $S \in \fs$ is contained in the interval $I_S$, i.e., $w(S) \in I_S$. We call $I_S$ the \emph{uncertainty interval} of set $S$. See~\Cref{fig:prelim:setsel-example-1} for an example.

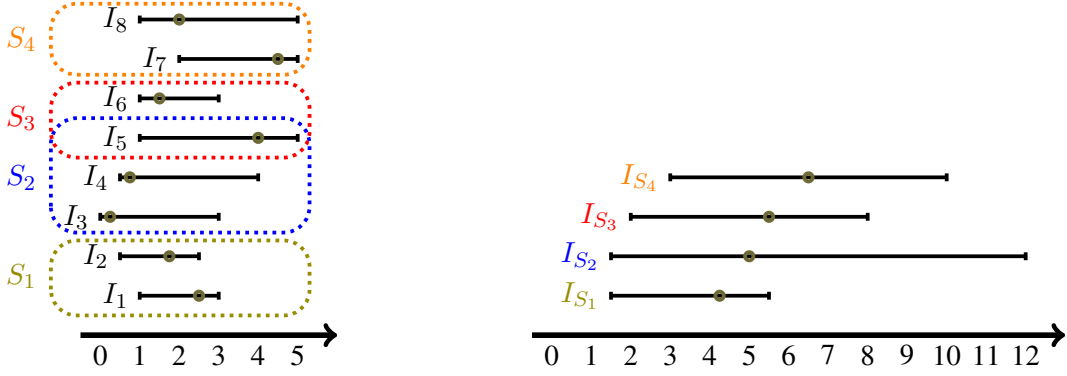
\begin{figure}[tb]
	\centering		
	\begin{subfigure}{0.4\textwidth}
		\begin{tikzpicture}[thick,line width = \lw, scale=0.525]
			\draw[->, line width = 0.75mm] (-0.5,-0.5) -- (6,-0.5);

			\node[] (l1) at (-2,1) {\textcolor{olive}{$S_1$}};
			\path [rounded corners=2ex] (-1.25, 0) [draw,dotted, color= olive] rectangle (5.3,1.9);
			\intervalr{$I_1$}{1}{3}{0.5}{2.5}
			\intervalr{$I_2$}{0.5}{2.5}{1.5}{1.75}

			\node[] (l1) at (-2,3.5) {\textcolor{blue}{$S_2$}};
			\path [rounded corners=2ex] (-1.25, 2.1) [draw,dotted, color= blue] rectangle (5.3,5);
			\intervalr{$I_3$}{0}{3}{2.5}{0.25}
			\intervalr{$I_4$}{0.5}{4}{3.5}{0.75}
			\intervalr{$I_5$}{1}{5}{4.5}{4}

			\node[] (l1) at (-2,5) {\textcolor{red}{$S_3$}};
			\path [rounded corners=2ex] (-1.25,4) [draw,dotted, color= red] rectangle (5.3,5.9);
			\intervalr{$I_6$}{1}{3}{5.5}{1.5}

			\node[] (l1) at (-2,7) {\textcolor{orange}{$S_4$}};
			\path [rounded corners=2ex] (-1.25,6.1) [draw,dotted, color= orange] rectangle (5.3,7.9);
			\intervalr{$I_7$}{2}{5}{6.5}{4.5}			
			\intervalr{$I_8$}{1}{5}{7.5}{2}
			
			\node[] (l1) at (0,-1) {0};
			\node[] (l1) at (1,-1) {1};
			\node[] (l1) at (2,-1) {2};
			\node[] (l1) at (3,-1) {3};
			\node[] (l1) at (4,-1) {4};
			\node[] (l1) at (5,-1) {5};
		\end{tikzpicture}  
	\end{subfigure}
	\begin{subfigure}{0.55\textwidth}
		\begin{tikzpicture}[thick,line width = \lw, scale=0.525]
			\draw[->, line width = 0.75mm] (-0.5,-0.5) -- (13,-0.5);
			\node[] (l1) at (0,-1) {0};
			\node[] (l1) at (1,-1) {1};
			\node[] (l1) at (2,-1) {2};
			\node[] (l1) at (3,-1) {3};
			\node[] (l1) at (4,-1) {4};
			\node[] (l1) at (5,-1) {5};
			\node[] (l1) at (6,-1) {6};
			\node[] (l1) at (7,-1) {7};
			\node[] (l1) at (8,-1) {8};
			\node[] (l1) at (9,-1) {9};
			\node[] (l1) at (10,-1) {10};
			\node[] (l1) at (11,-1) {11};
			\node[] (l1) at (12,-1) {12};

			\intervalr{$\textcolor{olive}{I_{S_1}}$}{1.5}{5.5}{0.5}{4.25}
			\intervalr{$\textcolor{blue}{I_{S_2}}$}{1.5}{12}{1.5}{5}
			\intervalr{$\textcolor{red}{I_{S_3}}$}{2}{8}{2.5}{5.5}
			\intervalr{$\textcolor{orange}{I_{S_4}}$}{3}{10}{3.5}{6.5}
			
			\phantom{				\path [rounded corners=2ex] (-1.25,6.1) [draw,dotted, color= orange] rectangle (5.3,7.9);}
			
		\end{tikzpicture}
	\end{subfigure}
	\caption{Instance for set selection under explorable uncertainty with intervals $\ui = \{I_1,I_2,\ldots, I_8\}$ and sets $\fs = \{S_1,S_2,S_3,S_4\}$ with $S_1 = \{I_1,I_2\}$, $S_2 = \{I_3,I_4,I_5\}$, $S_3 = \{I_4,I_5,I_6\}$ and $S_4 =  \{I_7,I_8\}$ (left) and the corresponding uncertainty intervals $I_{S_j}$ for the sets $S_j$ with $j \in \{1,\ldots,4\}$ (right). Green circles illustrate the precise values.}
	\label{fig:prelim:setsel-example-1}
\end{figure}

Since the intervals $(L_{S},U_{S})$  of the sets $S \in \fs$ can overlap, we might have to execute queries to determine the set of minimum value.
A query to an interval $I_i$ reveals the precise value $w_i$ and, thus, replaces both, $L_i$ and $U_i$, with $w_i$.
In a sense, a query to an $I_i \in S$ reduces the range $(L_S,U_S)$ in which $w(S)$ might lie by increasing $L_{S}$ by $w_i - L_i$ and decreasing $U_{S}$ by $U_i-w_i$. See~\Cref{fig:prelim:setsel-example-3} for an example.
We use $L_{S}$ and $U_{S}$ to refer to the initial limits and $L_{S}(Q)$ and $U_{S}(Q)$ to denote the limits of a set $S \in \fs$ \emph{after} querying a set of intervals~$Q \subseteq \ui$.

\begin{figure}[b]
	\begin{tikzpicture}[thick,line width = \lw, scale=0.85]
		\draw[black,dotted] (0, 0) node[black,anchor=east]{$I_i$}  -- (3, 0);
		\draw[black] (0, -0.1) -- (0, 0.1);
		\draw[black] (3, -0.1) -- (3, 0.1);
		\drawreal{1.25}{0}
		
		\draw[decoration= {brace, amplitude = 5 pt, aspect = 0.5}, decorate] (1.2,-0.25) -- (0,-0.25);
		\draw[decoration= {brace, amplitude = 5 pt, aspect = 0.5}, decorate] (3,-0.25) -- (1.3,-0.25);
		
		\node[] (l1) at (0.635,-0.75) {\small$w_i-L_i$};
		\node[] (l1) at (2.375,-0.75) {\small$U_i-w_i$};

		\draw[black] (6+1.25, 0.25)  -- (10-1.75, 0.25);
		\draw[black] (6+1.25, 0.25-0.1) -- (6+1.25, 0.25+0.1);
		\draw[black] (10-1.75, 0.25-0.1) -- (10-1.75, 0.25+0.1);
		
		\draw[black,dotted] (6, 0.25) node[black,anchor=east]{$I_{S_1}$} -- (10, 0.25);
		\draw[black] (6, 0.25-0.1) -- (6, 0.25+0.1);
		\draw[black] (10, 0.25-0.1) -- (10, 0.25+0.1);

		\draw[black] (7+1.25, -0.25)  -- (14-1.75, -0.25);
		\draw[black] (7+1.25, -0.25-0.1) -- (7+1.25, -0.25+0.1);
		\draw[black] (14-1.75, -0.25-0.1) -- (14-1.75, -0.25+0.1);
		
		\draw[black,dotted] (7, -0.25) node[black,anchor=east]{$I_{S_2}$}  -- (14, -0.25);
		\draw[black] (7, -0.25-0.1) -- (7, -0.25+0.1);	content...
		\draw[black] (14, -0.25-0.1) -- (14, -0.25+0.1);

		\draw[decoration= {brace, amplitude = 5 pt, aspect = 0.5}, decorate] (7+1.25,-0.5) -- (7,-0.5);
		\draw[decoration= {brace, amplitude = 5 pt, aspect = 0.5}, decorate] (14,-0.5) -- (14-1.75,-0.5);
		
		\node[] (l1) at (7+1.25*0.5,-1) {\small$w_i-L_i$};
		\node[] (l1) at (14-1.75*0.5,-1) {\small$U_i-w_i$};
	\end{tikzpicture}
	\caption{Example of how a query to an interval $I_i$ changes the intervals of two sets $S_1,S_2$ with $I_i \in S_1 \cap S_2$ in the set selection problem under explorable uncertainty.}
	\label{fig:prelim:setsel-example-3}
\end{figure}
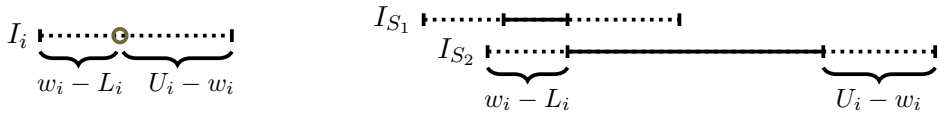

Let $w^* = \min_{S \in \fs} w(S)$ be the initially uncertain minimum value. 
To solve the problem, we have to adaptively query a set of intervals $Q$ until $U_{S^*}(Q) = L_{S^*}(Q) = w^*$ holds for some $S^* \in \fs$ and $L_S(Q) \ge w^*$ holds for all $S \in \fs$.
Only then, we know for sure that $w^*$ is indeed the minimum set value and that $S^*$ achieves this value. \Cref{fig:prelim:setsel-example-2} shows the structure of an instance that has been solved.
\begin{figure}[t]
	\centering		
	\begin{subfigure}{0.4\textwidth}
		\begin{tikzpicture}[thick,line width = \lw, scale=0.525]
			\draw[->, line width = 0.75mm] (-0.5,-0.5) -- (6,-0.5);

			\node[] (l1) at (-2,1) {\textcolor{olive}{$S_1$}};
			\path [rounded corners=2ex] (-1.25, 0) [draw,dotted, color= olive] rectangle (5.3,1.9);
			
			\node[] (l1) at (1.75,0.5) {$I_1$};
			\drawreal{2.5}{0.5}
			
			\node[] (l1) at (1,1.5) {$I_2$};
			\drawreal{1.75}{1.5}

			\node[] (l1) at (-2,3.5) {\textcolor{blue}{$S_2$}};
			\path [rounded corners=2ex] (-1.25, 2.1) [draw,dotted, color= blue] rectangle (5.3,5);
			\intervalr{$I_3$}{0}{3}{2.5}{0.25}
			\intervalr{$I_4$}{0.5}{4}{3.5}{0.75}
			
			\node[] (l1) at (3.25,4.5) {$I_5$};
			\drawreal{4}{4.5}
			
			\node[] (l1) at (-2,5) {\textcolor{red}{$S_3$}};
			\path [rounded corners=2ex] (-1.25,4) [draw,dotted, color= red] rectangle (5.3,5.9);
			\intervalr{$I_6$}{1}{3}{5.5}{1.5}

			\node[] (l1) at (-2,7) {\textcolor{orange}{$S_4$}};
			\path [rounded corners=2ex] (-1.25,6.1) [draw,dotted, color= orange] rectangle (5.3,7.9);

			\node[] (l1) at (3.75,6.5) {$I_7$};
			\drawreal{4.5}{6.5}
			
			\intervalr{$I_8$}{1}{5}{7.5}{2}
			
			\node[] (l1) at (0,-1) {0};
			\node[] (l1) at (1,-1) {1};
			\node[] (l1) at (2,-1) {2};
			\node[] (l1) at (3,-1) {3};
			\node[] (l1) at (4,-1) {4};
			\node[] (l1) at (5,-1) {5};
		\end{tikzpicture}  
	\end{subfigure}
	\begin{subfigure}{0.55\textwidth}
		\begin{tikzpicture}[thick,line width = \lw, scale=0.525]
			\draw[->, line width = 0.75mm] (-0.5,-0.5) -- (13,-0.5);
			\node[] (l1) at (0,-1) {0};
			\node[] (l1) at (1,-1) {1};
			\node[] (l1) at (2,-1) {2};
			\node[] (l1) at (3,-1) {3};
			\node[] (l1) at (4,-1) {4};
			\node[] (l1) at (5,-1) {5};
			\node[] (l1) at (6,-1) {6};
			\node[] (l1) at (7,-1) {7};
			\node[] (l1) at (8,-1) {8};
			\node[] (l1) at (9,-1) {9};
			\node[] (l1) at (10,-1) {10};
			\node[] (l1) at (11,-1) {11};
			\node[] (l1) at (12,-1) {12};

			\node[] (l1) at (3.5,0.5) {$\textcolor{olive}{I_{S_1}}$};
			\drawreal{4.25}{0.5}
			
			\intervalr{$\textcolor{blue}{I_{S_2}}$}{1.5+3}{12-1}{1.5}{5}
			\intervalr{$\textcolor{red}{I_{S_3}}$}{2+3}{8-1}{2.5}{5.5}
			\intervalr{$\textcolor{orange}{I_{S_4}}$}{3+2.5}{10-0.5}{3.5}{6.5}
			
			\phantom{				\path [rounded corners=2ex] (-1.25,6.1) [draw,dotted, color= orange] rectangle (5.3,7.9);}
			
		\end{tikzpicture}
	\end{subfigure}
	\caption{Instance of~\Cref{fig:prelim:setsel-example-1} after querying $Q = \{I_1,I_2,I_5,I_7\}$: Updated uncertainty intervals $\ui$ (left) and updated set uncertainty intervals (right).}
	\label{fig:prelim:setsel-example-2}
\end{figure}
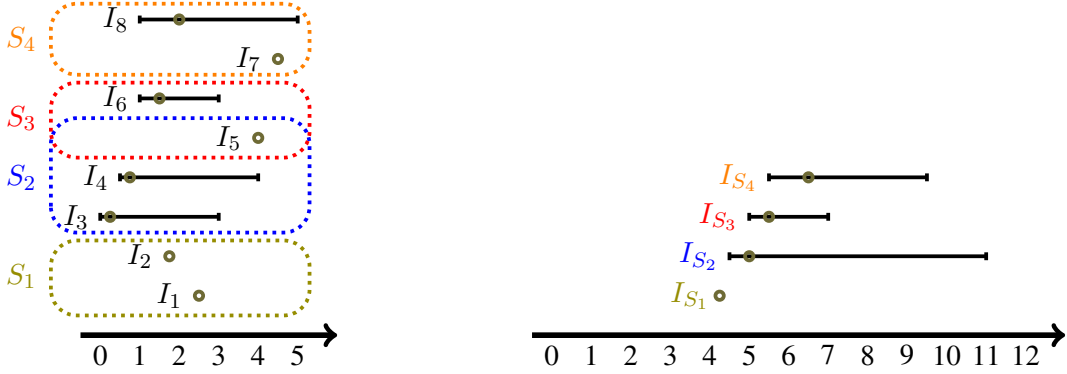
For an instance $(\ui, \fs)$ of \minset, the following integer linear program (ILP) with $a_i = w_i - L_i$ for all $I_i \in \ui$ and $b_S = w^* - L_S$ for all $S \in \fs$ formulates this problem:
\begin{equation}\tag{{\sc MinSetIP}}\label{eq:SetSelection}
	\begin{array}{lll}
		\min &\sum_{I_i \in \ui} x_i\\
		\text{s.t. }& \sum_{I_i \in S} x_i \cdot a_i \ge b_S &\forall S \in \fs\\
		& x_i \in \{0,1\}& \forall I_i \in \ui
	\end{array}
\end{equation}
Here, the variable $x_i$, $I_i \in \ui$, indicates whether interval $I_i$ is selected to be queried ($x_i = 1$) or not ($x_i = 0$) and our objective is to minimize the number of queries.

Observe that this ILP is a special case of the \emph{multiset multicover problem} (see, e.g.,~\cite{rajagopalan1998}).
If $a_i = w_i-L_i = 1$ for all $I_i \in \ui$ and $b_S = w^* - L_{S} = 1$ for all $S \in \fs$, then the problem is exactly the classical \setcover problem with $\ui$ corresponding to the \setcover \emph{sets} and $\fs$ corresponding to the \setcover \emph{elements}.

The optimal solution to~\eqref{eq:SetSelection} is the optimal query set for the  corresponding \minset instance; this is not hard to see but we also formally prove it in~\Cref{app:equivalence}. 
Under uncertainty however, the coefficients $a_i = w_i - L_i$ and right-hand sides $b_S = w^* - L_S$ are unknown to us. We only know that $a_i \in (L_i-L_i, U_i-L_i)= (0,U_i-L_i)$ as $a_i = (w_i-L_i)$ and $w_i \in (L_i,U_i)$. 
Only once we query an interval $I_i$, the precise value $w_i$ and, thus, the coefficient $a_i$ is revealed to us.
In a sense, to solve \minset under uncertainty, we have to solve~\eqref{eq:SetSelection} with uncertainty in the coefficients
and with irrevocable decisions.
For the rest of the paper, we interpret \minset under uncertainty in exactly that way: We have to solve~\eqref{eq:SetSelection} without knowing the coefficients in the constraints.
Whenever we \emph{irrevocably} add an interval $I_i$ to our solution (i.e., set $x_i$ to $1$), the information on the coefficients (in form of $w_i$) is revealed to us. 
Our goal is to add elements to our solution until it becomes feasible for~\eqref{eq:SetSelection}, and to minimize the number of added elements.
In this interpretation, the terms \enquote{querying an element} and \enquote{adding it to the solution} are interchangeable, and we use them as such.

Our main contribution is an algorithmic framework that exploits techniques for classical covering problems and adapts them to handle uncertainty in the coefficients $a_i$ and the right-hand sides $b_S$.
This framework not only allows us to obtain improved results for \minset under stochastic uncertainty but also to solve other covering problems with uncertainty in the constraints. 

\subsection{Our Results}

We design a polynomial-time algorithm for \minset under stochastic uncertainty with competitive ratio $\mathcal{O}(\frac{1}{\tau} \cdot \log^2 m)$, where $m$ is the number of sets (number of constraints in~\eqref{eq:SetSelection}) and parameter $\tau$ characterizes how \enquote{balanced} the distributions of values within the given intervals are. More precisely, $\tau = \min_{I_i \in \ui} \tau_i$ and $\tau_i$ is the probability that $w_i$ is larger than the center of $I_i$ (e.g., for uniform distributions $\tau = \frac{1}{2}$).
\jnew{All our results assume $\tau > 0$.}
This is the first stochastic result in explorable uncertainty concerning the sum of unknown values and it builds on new methods that shall be useful for solving more general problems in this field.
The ratio is independent of the number of elements, $n$. In particular for a small number of sets,~$m$, this is a significant improvement upon the adversarial lower bound of $n$~\cite{erlebach16cheapestset}.
Dependencies on parameters such as $\tau$ are quite standard and necessary
~\cite{MaeharaY20,Blum2020,GoemansV06,Vondrak07,Behnezhad2022}. For example, in~\cite{MaeharaY20} the upper bounds depend on the probability to draw the largest value of the uncertainty interval, which is an even stricter assumption that does not translate to open intervals.


We remark that the hidden constants in the performance bounds depend on the upper limits of the given intervals.
Assuming those to be constant is also a common assumption; see, e.g.,~\cite{maehara2020}.
Even greedy algorithms for covering problems similar to~\eqref{eq:SetSelection} \emph{without} uncertainty have such dependencies~\cite{rajagopalan1998,vazirani2001,dobson1982}.
While there exist non-greedy algorithms for covering problems without such dependencies~\cite{Kolliopoulos2002,Kolliopoulos2005}, it remains open whether they can be adjusted to the setting with uncertainty and, in particular, irrevocable decisions.

As \minset contains the classical \setcover problem, an approximation factor better than $\mathcal{O}(\log m)$ is unlikely, unless P$=$NP~\cite{Dinur2014}. 
We show that this holds also in the stochastic setting, even with uniform distributions. 
We further show that \jnew{$\frac{2}{\tau}$} is a lower bound for both problems under stochastic explorable uncertainty, even if the sets are pairwise disjoint. 
Hence, the dependencies on $\log m$ and \jnew{$\frac{2}{\tau}$} in our upper bounds are necessary. 

In the special case that all given sets are disjoint, we provide a simpler algorithm with competitive ratio $\frac{2}{\tau}$\jnew{, which matches the lower bound}. 
This is a gigantic improvement compared to the adversarial setting, where the lower bound of $n$ holds even for disjoint sets~\cite{erlebach16cheapestset}.

We remark that all our results for \minset translate to the maximization variant of the problem, where we have to determine the set of maximum value (cf.~\Cref{app:maxset}).

Algorithmically, we exploit the covering point of view to introduce a class of greedy algorithms that use the same basic strategy as the classical \setcover greedy algorithm~\cite{chvatal1979}.
However, we do not have sufficient information to compute and query an exact greedy choice under uncertainty as this choice depends on uncertain parameters.
Instead, we show that it is sufficient to query a small number of elements that together achieve a similar greedy value to the exact greedy choice.
If we do this repeatedly and the number of queries per iteration is small in expectation, then we achieve guarantees comparable to the approximation factor of a greedy algorithm with full information.
It is worth noting that this way of comparing an algorithm to the optimal solution is a novelty in explorable uncertainty as all previous algorithms for adversarial explorable uncertainty (\minset and other problems) exploit \emph{witness sets}. A witness set is a set of queries $Q$ such that each feasible solution has to query at least one element of $Q$, which allows to compare an algorithm with an optimal solution.

Our results translate to different covering problems under uncertainty. 
In particular, we consider the variant of~\eqref{eq:SetSelection}  under uncertainty with \emph{deterministic right-hand sides}. 
We give a simplified algorithm with improved competitive ratio~$\mathcal{O}(\frac{1}{\tau} \cdot \log m)$.
For a different balancing parameter, this holds even for the more general variant, where a variable can have different coefficients for the different constraints each with an individual uncertainty interval and distribution.  
For this problem, adding an element to the solution reveals \emph{all} corresponding coefficients.

\subsection{Further Previous Work}

Since \minset under uncertainty can be interpreted as both, a query minimization problem and a covering problem with uncertainty, we in the following summarize previous work from both fields.

\subsubsection{Previous Work on Query Problems}

For adversarial \minset under uncertainty, Erlebach et al.~\cite{Erlebach2016} show a (best possible) competitive ratio of $2d$, where $d$ is the cardinality of the largest set. In the lower bound instances, $d\in \Omega(n)$.
The algorithm repeatedly queries disjoint witness sets of size at most $2d$. This result was stated for the setting, in which it is not necessary to determine the value of the minimal set; if the value has to be determined, the bounds change to $d$. 

Further related work on \minset includes the result by~Maehara and Yamaguchi~\cite{maehara2020}, who consider packing ILPs with (stochastic) uncertainty in the cost coefficients, which can be queried.
They present a framework for solving several problems and bound the absolute number of iterations that it requires to solve them, instead of the competitive ratio.
However, we show in Appendix~\ref{app:yamaguchi} that their algorithm has competitive ratio $\Omega(n)$ for \minset under uncertainty, even for uniform distributions. Thus, it does not improve upon the adversarial lower bound.

Also Wang et al.~\cite{Wang22} consider selection-type problems in a somewhat related model. 
In contrast to our setting, they consider different constraints on the set of queries that, in a way, imply a budget on the number of queries. 
They 
solve optimization problems with respect to this budget, which has a very different flavor than our setting of minimizing the number of queries.

Furthermore, there is related work in a setting, where a query reveals the {\em existence} of entities instead of numeric values, e.g., the existence of edges in a graph,~c.f.~\cite{Blum2020,GoemansV06,Vondrak07}.
For example, Behnezhad et al.~\cite{Behnezhad2022} showed that vertex cover can be approximated within a factor of $(2+\epsilon)$ with only a constant number of queried edges per vertex. 
As edges define constraints, the result considers uncertainty only in the right-hand sides.

\subsubsection{Previous Work on Covering Problems with Uncertainty}

We continue by summarizing previous work on covering problems in different adversarial and stochastic settings.\todo{This is probably too long. We can decide later what to keep.}

In the \emph{online} version of \setcover~\cite{alon2003}, we are given a ground set of elements and a family of subsets of these elements. In contrast to offline \setcover, we do not necessarily have to cover all elements of the ground set. Instead, the members of the ground set that we do actually have to cover arrive online in an adversarial manner. Whenever an element arrives, we have to cover it by irrevocably adding a set containing the element to our solution, unless a previously added set already contains the element.
In a sense, online \setcover is a variant of~\eqref{eq:SetSelection} under uncertainty, where only the right-hand sides are uncertain in $\{0,1\}$ and all left-hand side coefficients are known and either one or zero. In contrast to \minset under uncertainty, the adversary for online \setcover is in a sense more powerful when selecting the right-hand sides as they do not depend on a common value $w^*$. Because of these differences, online \setcover has a very different flavor to \minset under uncertainty. The same holds for the stochastic version of online \setcover~\cite{Gupta2023,Grandoni2013}, where the subset of elements to be covered is drawn from a probability distribution.

A different stochastic variant of \setcover considers a \emph{two-stage} version of the problem~\cite{ShmoysS04}. In the first stage, we do not yet know which members of the ground set actually need to be covered. After the first stage, the elements to be covered are drawn from a probability distribution and in the second stage we have full knowledge of the elements to be covered. The crux of this two-stage variant is that adding sets to the solution in the first stage can be cheaper than adding them to the solution during the second stage. This again leads to a very different flavor than our setting.

While these \setcover variants consider uncertainty in the set of elements that need to be covered, Goemans and Vondr{\'{a}}k~\cite{Goemans2006} consider a variant where the elements to be covered are certain but there is uncertainty in which elements are covered by the sets. 
For each set, a vector describing the elements that are covered by the set is drawn according to a probability distribution.
This corresponds to a variant of~\eqref{eq:SetSelection}, where all right-hand sides are one but the left-hand side coefficients are uncertain in $\{0,1\}$.
Even in comparison to \minset under uncertainty with deterministic right-hand sides, there are several further difference besides the restriction of the coefficients to values in $\{0,1\}$.
For one,~\cite{Goemans2006} assumes access to the probability distributions. In particular, their algorithms are able to compute certain expected values. For our stochastic setting, we do not have sufficient information to compute expected values and the adversary still has some power in the selection of the unknown distributions as long as it respects the balancing parameter.
On the other hand, their \setcover variant allows some distributions that are not possible in \minset. In particular, an interval $I_i$ in \minset has the same coefficient $a_i = (w_i-L_i)$ in each constraint for a set $S$ with $I_i \in S$. Such a restriction does not exist in the problem considered in~\cite{Goemans2006}.
This in a sense makes their problem incomparable to \minset under uncertainty.
Furthermore,~\cite{Goemans2006} analyzes the approximation ratio instead of the competitive ratio. That is, they compare the expected objective value of an algorithm against the expected objective of the best possible algorithm instead of the expected optimum. 
To that end, they give an $m$-approximation for the stochastic \setcover variant. If sets can be added to the solution multiple times while each time drawing a new realization from the same distribution, they give a $\mathcal{O}(\log m)$-approximation. 

Besides related work on stochastic \setcover variants, there is previous related work on the more general \emph{(stochastic) submodular covering problem} (cf., e.g.,~\cite{Wolsey82,AgarwalAK19,Ghuge2021}). In the \emph{submodular covering problem}, we are given a ground set of elements $E$ and a submodular function $f\colon 2^E \rightarrow \mathbb{N}_+$. The goal is to find a subset $S \subseteq E$ of minimum cardinality such that $f(S) = f(E)$. 
This non-stochastic submodular covering problem contains offline~\minset~\cite{Wolsey82}, i.e., \eqref{eq:SetSelection} with full knowledge of all coefficients and right-hand sides. To see this, consider an instance $(\ui,\fs)$ of \minset. We can interpret the intervals as the ground set of elements, i.e., $E = \ui$ and use the submodular function $f(Q) = \sum_{S \in \fs} \min\{\sum_{I_i \in S \cap Q}  w_i - L_i, w^* - L_S\}$ for $Q \subseteq E$. Then, $f(E)$ is the sum of right-hand sides of~\eqref{eq:SetSelection} and $f(Q) = f(E)$ holds for a subset $Q \subseteq \ui = E$ if and only if $Q$ is feasible for~\eqref{eq:SetSelection}.
The best-known algorithm for the submodular covering problem achieves an approximation ratio of $\mathcal{O}(\log(f(E)))$~\cite{Wolsey82} and no polynomial-time algorithm can be better unless P=NP~\cite{Dinur2014}.

In the \emph{stochastic submodular covering problem}, we are given random variables $X_1,\ldots,X_n$ that independently realize to subsets of $E$ according to known probability distributions.
The task is to sequentially and irrevocably add random variables $X_i$ to the solution $\mathcal{X}$ until $f(\bigcup_{X_i \in \mathcal{X}} X_i) = f(E)$.
Whenever a random variable $X_i$ is added to the solution, the realization of the variable is revealed.
While this general setting is similar to \minset under uncertainty, there are some differences. 
In the stochastic submodular covering problem, the value $f(\bigcup_{X_i \in \mathcal{X}} X_i)$ only depends on the realizations of the random variables in $\mathcal{X}$. For the submodular function defined above for a \minset instance, the function value $f(Q)$ depends also on the uncertain $w^*$ and, therefore, on elements outside of $Q$. Furthermore, as the intervals $\ui$ in a \minset instance are continuous, modeling them as a stochastic submodular covering instance would require some form of discretization.
Independent of these differences, all results on the stochastic submodular covering problem (to our knowledge) assume known distributions and actively use them, which is in contrast to our stochastic setting. Furthermore, all these results analyze the approximation ratio instead of the competitive ratio. Thus, existing results for the stochastic submodular covering problem cannot directly be applied to our stochastic setting. 
This also holds for a range of problem variants that have been considered in the literature (see, e.g., \cite{GolovinK11,KambadurNN17,ImNZ16,DeshpandeHK16,NavidiKN20}).

\subsection{Outline}

To start the paper, we, in~\Cref{sec:Disjoint}, consider the special case of \minset under uncertainty with pairwise disjoint sets.  We give a lower bound of $\frac{2}{\tau}$ and a matching upper bound on the competitive ratio for this special case. These bounds nicely illustrate the challenges caused by the uncertainty and the techniques that we use to tackle them, also later on for the general problem. 

Afterwards, in~\Cref{sec:framework}, we move on to the general \minset and discuss the hardness of approximation as well as approximations of the offline problem variant. Based on observations for the offline problem, we introduce an algorithmic framework that can be used to solve \minset under uncertainty and other covering problems with uncertainty in the constraints. 

For the remaining paper, we show how to implement the framework for \minset under uncertainty with deterministic right-hand sides (\Cref{sec:mincover}), for more general covering problems under uncertainty with deterministic right-hand sides (\Cref{sec:mincover}) and, finally, for the general \minset (\Cref{sec:minset}). Using these implementations and our observations for the special case of disjoint sets, we prove our algorithmic results.

\section{Disjoint \minset}
\label{sec:Disjoint}

Consider the special case of \minset where all sets are pairwise disjoint, i.e., $S \cap S' = \emptyset$ for all $S,S' \in \fs$ with $S \not= S'$. 
We call this special case \emph{disjoint \minset}.
Disjoint \minset is of particular interest as it gives 
lower bounds for several problems under adversarial explorable uncertainty, cf.~\cite{erlebach16cheapestset,meissner18querythesis}.
To illustrate the challenges posed by having stochastic uncertainty in the input, we give the following lower bound.
Recall that the balancing parameter is defined as $\tau = \min_{I_i \in \ui} \tau_i$, where $\tau_i$ is the probability that $w_i$ is larger than the center of $I_i$. We use $\ALG$ and $\OPT$ to refer to an algorithm and an optimal solution, respectively. Slightly abusing the notation, we use the same terms to refer also to the corresponding numbers of queries.

\jnew{First, we show the following lower bound that even holds for known probability distributions. Afterwards, we prove a slightly stronger bound exploiting unknown distributions.}\todo{Maybe just remove the weaker bound?}

\begin{theorem}
	\label{thm:lb:known-dists}
	No deterministic algorithm for \minset under uncertainty has a competitive ratio better than $\frac{1}{\tau}$, even if all given sets are pairwise disjoint and \jnew{the distributions are known.}
\end{theorem}

\begin{proof}
	Consider an instance with the set of uncertainty intervals $\ui = \{I_0,I_1,\ldots,I_n\}$ with $I_0 = \{0.65\}$ and $I_i = (0,1)$ for all $i \not= 0$, and sets $\fs = \{S_1,S_2\}$ with $S_1 = \{I_0\}$ and $S_2 = \ui \setminus \{I_0\}$. 
	See~\Cref{fix:prelim:lb-setsel} for an illustration.
	Define the distributions $d_{i}$ with $i\not= 1$ as $d_{i}(a) = (1-\tau)$ if $a = \epsilon$,  $d_{i}(a) = \tau$ if $a = 0.7$ and $d_i(a) = 0$ otherwise, for some infinitesimally small $\epsilon > 0$.
	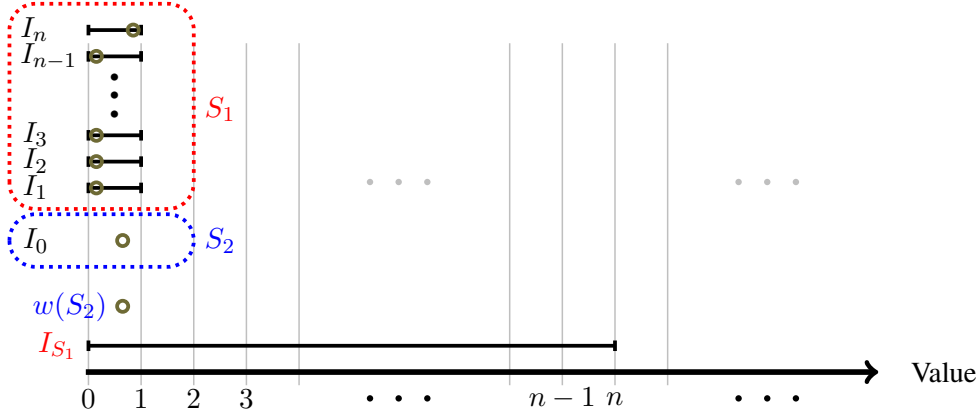
\begin{figure}[t]
		\centering	
		\begin{tikzpicture}	[thick,line width = \lw, scale=0.7]
			
			\draw[->, line width = 0.75mm] (0-0.5mm,-1) -- (15,-1);
			\node[] (l1) at (16.25,-1) {Value};
			
			\node[] (l1) at (0,-1.5) {$0$};
			\node[] (l1) at (1,-1.5) {$1$};
			\node[] (l1) at (2,-1.5) {$2$};
			\node[] (l1) at (3,-1.5) {$3$};
			
			\node[] (l1) at (6,-1.5) {\Huge$\dots$};
			\node[] (l1) at (13,-1.5) {\Huge$\dots$};

			\node[] (l1) at (9,-1.5) {$n-1$};
			\node[] (l1) at (10,-1.5) {$n$};

			\begin{scope}[on background layer]
				\draw[line width = 0.2mm,,lightgray] (0,-1.25) -- (0,5.25);
				\draw[line width = 0.2mm,,lightgray] (1,-1.25) -- (1,5.25);
				\draw[line width = 0.2mm,,lightgray] (2,-1.25) -- (2,5.25);
				\draw[line width = 0.2mm,,lightgray] (3,-1.25) -- (3,5.25);
				\draw[line width = 0.2mm,,lightgray] (4,-1.25) -- (4,5.25);
				
				\node[lightgray] (l1) at (6,2.6125) {\Huge$\dots$};
				
				\draw[line width = 0.2mm,,lightgray] (8,-1.25) -- (8,5.25);
				\draw[line width = 0.2mm,,lightgray] (9,-1.25) -- (9,5.25);
				\draw[line width = 0.2mm,,lightgray] (10,-1.25) -- (10,5.25);
				\draw[line width = 0.2mm,,lightgray] (11,-1.25) -- (11,5.25);
				
				\node[lightgray] (l1) at (13,2.6125) {\Huge$\dots$};
			\end{scope}

			\path [rounded corners=2ex] 
			(-1.5,2.1) [draw,dotted, color= red] rectangle (2,6);
			\node[] (l1) at (2.5,4) {\textcolor{red}{$S_1$}};
			
			\intervalr{$I_{n\phantom{-1}}$}{0}{1}{5.5}{0.85}
			\intervalr{$I_{n-1}$}{0}{1}{5}{0.15}
			\node[] (l1) at (0.5,4.25) {\tiny$\vertDOTS$};
			\intervalr{$I_{3\phantom{-1}}$}{0}{1}{3.5}{0.15}
			\intervalr{$I_{2\phantom{-1}}$}{0}{1}{3}{0.15}
			\intervalr{$I_{1\phantom{-1}}$}{0}{1}{2.5}{0.15}
			
			\interval{$\textcolor{red}{I_{S_1}}$}{0}{10}{-0.5}
			
			\drawreal{0.65}{1.5}
			\node[] (l1) at (-1,1.5) {$I_0$};
			\node[] (l1) at (2.5,1.5) {\textcolor{blue}{$S_2$}};
			\path [rounded corners=2ex] (-1.5,1) [draw,dotted, color= blue] rectangle (2,2);
			
			\drawreal{0.65}{0.25}
			\node[] (l1) at (-0.35,0.25) {$\textcolor{blue}{w(S_2)}$};
			
		\end{tikzpicture}
		\caption{Lower bound example for the set selection problem under explorable uncertainty consisting of the intervals $\ui = \{I_0,\ldots,I_n\}$ and the sets $\fs = \{S_1,S_2\}$ with $S_1 = \{I_1,\ldots,I_n\}$, $S_2 = \{I_0\}$, $I_0 = \{0.65\}$ and $I_i=(0,1)$ for $i\in \{1,\ldots,d\}$.}
		\label{fix:prelim:lb-setsel}
	\end{figure}

	If there exists some $I_i \in S_2$ with $w_i = 0.7$, then $\opt=1$ as a query to that interval already proves that $S_1$ is the set of minimum value since $w(S_1) = 0.65 < 0.7$.
	Otherwise, $\opt=n$.
	Therefore, $\EX[\opt] = (1 - (1-\tau)^{n}) + (1-\tau)^{n} \cdot n$ and $\lim_{n \rightarrow \infty} \EX[\opt] = 1$.

	Since $I_0$ is trivial and all $I_i$ with $i \not= 0$ are identical with the same distribution, each deterministic algorithm $\alg$ will just query the elements of $S_2$ in some order until it either reaches an $I_i$  with $w_i = 0.7$ or has queried all intervals.
	This implies that $\alg$ is a geometrical distribution with success probability $\tau$ and, therefore, $\EX[\alg] =\min\{n,\frac{1}{\tau}\}$.
	For $n$ towards infinity, we get
	$$
	\lim_{n \rightarrow \infty} \frac{\EX[\alg]}{\EX[\opt]} =\frac{1}{\tau}.
	$$	
\end{proof}

\jnew{Next, we show a slightly stronger bound for unknown distributions. The lower bound instance heavily exploits that, for unknown distributions, the adversary still has some power when selecting the probability distributions.}

\begin{theorem}
	\label{thm:lb:unknown-dists}
	\jnew{No deterministic algorithm for \minset under uncertainty has a competitive ratio better than $\frac{2}{\tau}$, even if all given sets are pairwise disjoint.}
\end{theorem}

\begin{proof}
	\jnew{Consider the same instance as in the proof of~\Cref{thm:lb:known-dists} but with different, now unknown distributions. Since the distributions are unknown, an algorithm cannot distinguish the intervals $I_1,\ldots,I_n$ even if they have different distributions. This means that the adversary still has some power and can set the distributions in a worst-case manner for the algorithm, as long as the distributions respect balancing parameter~$\tau$.
		
		To that end, consider a fixed value $\tau$ and an arbitrary deterministic algorithm $\alg$. As $\alg$ cannot distinguish the intervals $I_1,\ldots,I_n$, we can assume w.l.o.g.~that it queries the intervals in order of their indices until the instance is solved. For all $0 < i < n$, the adversary sets the distribution to $d_i(a) = \tau$ for $a=0.51$, $d_i(a)= (1-\tau)$ for $a=\epsilon$ and $d_i = 0$ otherwise, for some infinitesimally small $\epsilon > 0$. Finally, the adversary sets distribution $d_n$ to $d_n(a) = 1$ for $a=0.7$ and $d_n(a)=0$ otherwise. These distributions clearly respect the balancing parameter $\tau$.
		
		For these distributions, we always have $w(S_2) > w(S_1)$ as $w_n > w(S_1)$ holds with a probability of one. Thus, every algorithm has to query until the lower limit of set $S_2$ increases to a value of at least $w(S_1)$. For $\alg$, this is the case once it found two intervals $I_i$ with $i < n$ and $w_i = 0.51$ or once it queries interval $I_n$ in case no two such intervals exist. Thus, the expected query cost is $\EX[\alg] = \min\{\frac{2}{\tau},n\}$. The optimal solution on the other hand only queries $I_n$ and is done after a single query. Therefore, $\EX[\opt] = 1$ and the competitive ratio of $\alg$ is at least $\min\{\frac{2}{\tau},n\}$. We can conclude the theorem by picking a sufficiently large value for $n$.
	}
\end{proof}

We continue by giving a quite simple algorithm for disjoint \minset~\jnew{that matches the lower bound of~\Cref{thm:lb:unknown-dists}.} 

In disjoint \minset, each $I_i$ occurs in exactly one constraint for one set $S$ in the corresponding~\eqref{eq:SetSelection}. 
Thus, each set $S$ defines a disjoint subproblem and
the optimal solution $\OPT$ of the 
instance is the union of 
optimal solutions 
for the 
subproblems.
The optimal solution for a subproblem $S$ is to query the elements of $I_i \in S$ in order of non-decreasing $(w_i-L_i)$ until the sum of those coefficients is at least $(w^* - L_{S})$.

Under uncertainty, we adapt this strategy and query in order of non-decreasing $(U_i-L_i)$.
While this does not guarantee that we query the interval with maximum $(w_j-L_j)$ in $S$, it gives us a probability of $\tau$ to query an interval $I_i$ such that $(w_i-L_i)$ is at least half the maximum $(w_j-L_j)$. We will prove that this is sufficient to achieve the guarantee.
Since we do not know $w^*$, we do not know when to stop querying in a subproblem.
We 
handle this by only querying in the set $S$ of minimum current lower limit as the subproblem for this set is clearly not yet solved. 
Algorithm~\ref{alg:disjoint} formalizes this approach.

\begin{algorithm}[tb]
	\KwIn{Instance of \minset under uncertainty with pairwise disjoint sets.}
	$Q \gets \emptyset$\; 
	\While{the problem is not solved}{
		$S_{\min} \gets \arg\min_{S \in \fs} L_{S}(Q)$\;
		\Repeat{$w_i-L_i \ge \frac{1}{2} \cdot (U_i-L_i)$ or $S_{\min}$ has been completely queried}{
			$I_i \gets \arg\max_{I_j \in S_{\min}\setminus Q} U_j-L_j$;\ 
			Query $I_i$;\  $Q \gets Q \cup \{I_i\}$\;
		}
	}
	\caption{Algorithm for disjoint \minset under uncertainty.}
	\label{alg:disjoint}
\end{algorithm}

\begin{restatable}{theorem}{thmDisjoint}
	\label{thm:disjoint}
	There is an algorithm for disjoint \minset under uncertainty with competitive ratio at most $\frac{2}{\tau}$.
\end{restatable}

\begin{proof}
	Consider a fixed realization of values $w_i$ and the corresponding~\eqref{eq:SetSelection} instance.
	For each $S \in \fs$, a feasible solution $Q$ must satisfy $\sum_{I_i \in S \cap Q} (w_i-L_i) \ge (w^* - L_{S})$.
	This implies $|Q \cap S| \ge |P_S^*|$ for the minimum cardinality prefix $P_S^*$ of $I_1,\ldots,I_k$ with $\sum_{I_i \in P_S^*} (w_i-L_i) \ge w^* - L_{S}$, where $S = \{I_1,\ldots,I_k\}$ and all $I_i$ are indexed by non-increasing $w_i - L_i$. 
	As the sets are disjoint, we get $\OPT = \sum_{S \in \fs} |P_S^*|$.
	
	Using this, we show that Algorithm~\ref{alg:disjoint} satisfies the theorem.
	To that end, let $X_j$ be a random variable denoting the number of queries in iteration $j$ of the outer while-loop of Algorithm~\ref{alg:disjoint} and let $Y_j$ be an indicator variable indicating whether iteration $Y_j$ is actually executed ($Y_j=1$) or not ($Y_j = 0$).
	
	We prove the theorem by separately showing  $\sum_j \PR[Y_j = 1] \le 2 \cdot \EX[\OPT]$ and $\EX[X_i \mid Y_j = 1] \le \frac{1}{\tau}$.
	Since 
	\begin{align*}
		\EX[\ALG] &= \sum_j \EX[X_j] = \sum_j \PR[Y_j = 0] \EX[X_j \mid  Y_j = 0] + \sum_j \PR[Y_j = 1] \EX[X_j \mid Y_j = 1]
		\\ &= \sum_j \PR[Y_j = 1] \EX[X_j \mid  Y_j = 1]
	\end{align*}
	follows from $\EX[X_j \mid  Y_j = 0] = 0$ and the law of total expectations, the two inequalities imply the theorem.
	
	Note that $\sum_j \PR[Y_j = 1]$ is just the expected number of iterations of the algorithm. Thus, if we show for each realization of precise values that the number of iterations is at most  $2 \cdot \OPT$, we directly get $\sum_j \PR[Y_j = 1] \le 2 \cdot \EX[\OPT]$.
	
	Consider a fixed realization.
	For each $S$, let $h_S$ denote the number of iterations with $S_{\min} = S$.
	We claim that $h_S \le 2 \cdot |P_S^*|$. 
	Then, $\OPT \ge \sum_{S_\ell \in \fs} |P_S^*|$ implies $\sum_j \PR[Y_j = 1] \le 2 \cdot \EX[\OPT]$.
	
	Let $j$ be an iteration with $S_{\min} = S$, let $G_j$ denote the queries of this iteration and let $Q_j$ denote the set of all previous queries.
	Observe that $P_S^* \setminus Q_j \not= \emptyset$.
	Otherwise, the definition of $P_S^*$ would either imply that the lower limit of $S$ after querying $Q_j$ is larger than $w^*$, which contradicts $S_{\min} = S$, or 
	that the lower limit is equal to $w^*$, which  implies that the problem is already solved.
	
	We argue that $\sum_{I_i \in G_j} w_i - L_i \ge \frac{1}{2} \cdot \max_{I_i \in S \setminus Q_j} w_i-L_i$.
	In case that interval $I_i = \arg\max_{I_i \in S \setminus Q_j} w_i-L_i$ is contained in $G_j$, the inequality clearly holds.
	Otherwise, let $I_{i'}$ denote the last element that is queried in the iteration. 
	Then, $w_{i'}-L_{i'} \ge \frac{1}{2} \cdot (U_{i'} - L_{i'}) \ge  \frac{1}{2} \cdot \max_{I_i \in S \setminus Q_j} w_i-L_i$, where the first inequality holds as $I_{i'}$ is the last query of the iteration and the second inequality holds by the order in which the elements of $S$ are queried by the algorithm.
	Thus, this last interval $I_{i'}$ alone satisfies the inequality.
	
	The inequality suffices to conclude that $\sum_{I_i \in Q_j \cap S} w_i - L_i \ge \sum_{I_i \in P_S^*} w_i - L_i$ holds after at most $2 \cdot |P_S^*|$ iterations with $S_{\min} = S$.
	As $S_{\min} \not= S$ holds for all following iterations, the claim $h_S \le 2 \cdot |P_S^*|$ follows.
	Next, we show the second inequality, $\EX[X_j \mid Y_j = 1] \le \frac{1}{\tau}$.
	If an iteration $j$ of the outer while-loop is executed ($Y_j = 1$), the repeat statement queries intervals $I_i$ until either $(w_i-L_i) \ge \frac{1}{2} \cdot (U_i-L_i)$ or $S_{\min} \subseteq Q$.
	Thus, it terminates at the latest when it finds an $I_i$ with $(w_i-L_i) \ge \frac{1}{2} \cdot (U_i-L_i)$.
	The number of queries until such an interval occurs is described by a geometric distribution with success probability at least $\tau$.
	So, in expectation, this number is at most $\frac{1}{\tau}$ and
	we can conclude $\EX[X_j \mid Y_j = 1] \le \frac{1}{\tau}$.	
\end{proof}

\jnew{We remark that~\Cref{thm:lb:unknown-dists,thm:disjoint} imply that, even with full knowledge of the distributions, the competitive ratio for disjoint \minset cannot be improved by more than a factor of two compared to the ratio with unknown distributions.}

\section{Algorithmic framework}
\label{sec:framework}

In the previous section, we have seen an algorithm for disjoint \minset under uncertainty with a tight competitive ratio. The key observation that allowed us to achieve that ratio was the simple characterization of an (offline) optimal solution. 
In this section, we consider the offline variant of the general \minset and give inapproximability results that prevent such simple characterizations for optimal solutions of the general problem.
Thus, we need alternative algorithms and, based on observations for the offline problem, present an algorithmic framework that can be used to solve \minset and other covering problems under uncertainty.

\subsection{Offline Problems and Hardness of Approximation}

We refer to the problem of solving~\eqref{eq:SetSelection} with full knowledge of the precise values $w_i$ (and $w^*$) as \emph{offline} problem.
This means that we have full knowledge of all coefficients of the ILP.
For \minset under uncertainty, we say that a solution is optimal, if it is an optimal solution for the corresponding offline problem.
We use $\OPT$ to refer to an optimal solution and, slightly abusing the notation, to its objective value.

Offline \minset contains the classical \setcover problem and, thus, it is as hard to approximate.
This result transfers to the stochastic setting, even for uniform distributions. 
Results by Dinur and Steurer~\cite{Dinur2014} imply the following, as we formally prove in Appendix~\ref{app:inapprox}.

\begin{restatable}{theorem}{inapprox}
	\label{thm:hardness}
	For any fixed $\alpha > 0$, it is NP-hard to compute a query strategy that is $(1-\alpha) \cdot \ln m$-competitive for \minset under uncertainty even if the precise value $w_i$ of each $I_i$ is drawn independently and uniformly at random from $(L_i,U_i)$.
	The same inapproximability holds also for offline \minset.
\end{restatable}

On the positive side, we can approximate offline \minset by adapting covering results (see, e.g.,\cite{chvatal1979,dobson1982,rajagopalan1998,Kolliopoulos2002,Kolliopoulos2005}).
In particular, we want to use greedy algorithms that iteratively and irrevocably add elements to the solution that are selected by a certain greedy criterion.
Recall that \enquote{adding an element to the solution} corresponds to both, setting the variable  $x_i$ of an interval $I_i \in \ui$ in~\eqref{eq:SetSelection} to one and querying $I_i$.
While we are technically not restricted to greedy algorithms when solving \emph{offline} \minset, our goal is to later on generalize the offline algorithm to the setting with uncertainty and irrevocable decisions.
\jnew{Hence, greedy algorithms seem to be a suitable choice.}

Since the greedy criterion for adding an element depends on previously added elements, we define a version of~\eqref{eq:SetSelection} that is parametrized by the set $Q \subseteq \ui$ of elements that have already been added to the solution and adjust the right-hand sides to the remaining covering requirement after adding $Q$.
Recall that $a_i = w_i - L_i$ and $b_S = w^* - L_S$.
Here, $b_S(Q) =  \max\{b_S - \sum_{I_i \in Q\cap S} a_i ,0\}$ and $b(Q) = \sum_{S \in \fs} b_S(Q)$.

\begin{equation}
	\tag{{\sc MinSetIP-Q}}
	\label{eq:SetSelection:param}
	\begin{array}{llll}
		\min &\sum_{I_i \in \ui\setminus Q} x_i\\
		\text{s.t. }& \sum_{I_i \in S\setminus Q} x_i \cdot a_i &\ge b_S(Q) &\forall S \in \fs\\
		& x_i & \in \{0,1\}& \forall I_i \in \ui \setminus Q\\
	\end{array}
\end{equation}
Based on this ILP, we adjust the algorithm by Dobson~\cite{dobson1982} for the multiset multicover problem to our setting (cf.~Algorithm~\ref{alg_verification_covering_1}).
The algorithm scales the coefficients such that all non-zero left-hand side coefficients are at least $1$.
We refer to such instances as \emph{scaled}.
Then it greedily adds the element to the solution that reduces the right-hand sides the most, i.e., the interval $I_i \in \ui\setminus Q$ that maximizes $\gc(Q,I_i) = b'(Q) - b'(Q\cup\{I_i\})$ ($a'$ and $b'$ indicate scaled coefficients). 
For a subset $G \subseteq \ui$, we define $\gc(Q,G) = b'(Q) - b'(Q \cup G)$. 

After $b_S'(Q) < 1$ for all $S \in \fs$, we can exploit that all scaled non-zero coefficients $a'_i$ are at least one. 
This means that adding an element $I_i \in \ui\setminus Q$ satisfies all remaining constraints of sets $S$ with $I_i \in S$.
Thus, the remaining problem reduces to a \setcover instance, which can be solved by using the classical greedy algorithm by Chvatal~\cite{chvatal1979}.
This algorithm greedily adds the element $I_i \in \ui\setminus Q$ that maximizes $\gs(Q,I_i) = A(Q) - A(Q \cup \{I_i\})$ with $A(Q) = |\{S \in \fs \mid b'_S(Q) > 0\}|$, i.e., the element that satisfies the largest number of constraints that are not already satisfied by $Q$.
For a subset $G \subseteq \ui$, we define $\gs(Q,G) =  A(Q) - A(Q \cup G)$.

During the course of this paper, we refer to $\gc(Q,I_i)$, $\gs(Q,I_i)$, $\gc(Q,G)$ and $\gs(Q,G)$ as the \emph{greedy values} of $I_i$ and $G$, respectively.

\begin{theorem}[Follows from Dobson~\cite{dobson1982}]
	\label{thm:dobson}
	Algorithm~\ref{alg_verification_covering_1} is a polynomial-time $\mathcal{O}(\log m)$-approximation for offline \minset. The precise approximation factor is $\grset(\gamma) = \lceil\ln(\gamma \cdot m\cdot \max_S (w^*-L_{S}))\rceil + \lceil\ln(m)\rceil$ with $s_{\min} =  \min_{I_i \in \ui \colon (w_i-L_i) >0} (w_i-L_i)$, $\gamma = 1/s_{\min}$ and $m =  |\fs|$.
\end{theorem}

During the remaining course of the paper, we will state the competitive ratios of our algorithms in terms of $\grset$.
To that end, define $\grsetu(\gamma) = \lceil\ln(\gamma \cdot m\cdot \max_{S,S'} (U_{S}-L_{S'})\rceil + \lceil\ln(m)\rceil$, which is an upper bound on $\grset(\gamma)$.
Under uncertainty, we compare against $\grsetu$ to avoid the random variable $w^*$. For constant $U_i$'s, $\grsetu$ and $\grset$ are asymptotically the same.

We remark again that the approximation ratio of Algorithm~\ref{alg_verification_covering_1} has dependencies on the numerical input parameter $s_{\min}$ and $\max_S(w^*-L_S)$. While there exist algorithms that achieve an approximation ratio of $\mathcal{O}(\log m)$ for the offline problem without such dependencies~\cite{Kolliopoulos2002,Kolliopoulos2005}, these algorithms are not greedy and it remains open whether there exist algorithms with this improved ratio that execute irrevocable decisions, even with full knowledge of the coefficients. Thus, we consider Algorithm~\ref{alg_verification_covering_1} and aim at extending it for the setting under uncertainty.

\begin{algorithm}[tb]
	\KwIn{An instance of offline~\minset, i.e., an instance of~\eqref{eq:SetSelection}}
	$s_{\min} = \min_{I_i \in \ui \colon a_i > 0} a_i$;
	$\forall S \in \fs \colon b'_S = \frac{b_S}{s_{\min}}$; 
	$\forall I_i \in \ui \colon a'_i = \frac{a_i}{s_{\min}}$\;
	\While{$\exists S \in \fs \colon b'_S(Q) \ge 1$}{			
		$I_i \gets \arg\max_{I_j \in \ui\setminus Q} \gc(Q,I_j)$; Query $I_i$; $Q \gets Q \cup \{I_i\}$\;			
	}
	\While{the problem is not solved}{
		$I_i \gets \arg\max_{I_j \in \ui\setminus Q} \gs(Q,I_j)$; Query $I_i$; $Q \gets Q \cup \{I_j\}$\;
	}
	\caption{Greedy algorithm by Dobson~\cite{dobson1982} for offline~\minset.}
	\label{alg_verification_covering_1}
\end{algorithm}

\subsection{Algorithmic framework}

We introduce our algorithmic framework that we use to solve \minset under uncertainty.
Ideally, we would like to apply the offline greedy algorithm.
However, since the coefficients $a_i = w_i-L_i$ and $b_S = w^*-L_S$ are unknown, we cannot apply Algorithm~\ref{alg_verification_covering_1} to solve \minset under uncertainty as we cannot compute the element that maximizes the greedy value $\gc$ or $\gs$.


While we cannot precisely compute the greedy choice, our strategy is to \emph{approximate} it and to show that approximating the greedy choice is sufficient to obtain the desired guarantees.
To make this more precise, consider an \emph{iterative} algorithm for~\eqref{eq:SetSelection}, i.e., an algorithm that iteratively adds pairwise disjoint subsets $G_1,\ldots,G_h$ of $\ui$ to the solution.
For each $j$, let $Q_j = \bigcup_{1 \le j' \le j-1} G_{j'}$, i.e., $Q_j$ contains the  elements that have been added to the solution before $G_j$.
If the combined greedy value of $G_j$ is within a factor of $\alpha$ to the best greedy value for the problem instance \emph{after} adding $Q_j$, then we say that $G_j$ $\alpha$-approximates the greedy choice.
The following technical definition makes this more precise and the subsequent lemma connects the definition to the actual greedy values while taking into account that there are two different greedy values $\gc$ and $\gs$ (cf.~Algorithm~\ref{alg_verification_covering_1}). 

\begin{definition}
	\label{def:greedyalpha}
	For a \eqref{eq:SetSelection} instance with scaled coefficients and optimal solution $\OPT$,
	let $\alpha \in \mathbb{R}_{\ge 1}$ and consider the corresponding instance of \eqref{eq:SetSelection:param} for some $Q \subseteq \ui$.
	A set $G \subseteq \ui\setminus Q$ \emph{$\alpha$-approximates} the current greedy choice after adding $Q$ if either 
	\begin{enumerate}
		\item $A(Q\cup G) \le (1-\frac{1}{\alpha \cdot \OPT}) \cdot A(Q)$ or
		\item $b'(Q) \ge 1$ and $b'(Q\cup G) \le (1-\frac{1}{\alpha \cdot \OPT}) \cdot b'(Q)$.
	\end{enumerate}	
\end{definition}

Intuitively, the two conditions of the following lemma seem like a more appropriate definition of approximating a greedy choice. 
While the conditions of the lemma imply that the definition above is satisfied, in our proofs it will sometimes be easier to directly show that the definition is satisfied, without using the lemma.
Therefore, we use the more technical~\Cref{def:greedyalpha} \jnew{but the lemma captures the intuition behind the definition.}

\begin{restatable}{lemma}{lemGreedyAlpha}
	\label{lem:greedy:alpha} For a scaled instance of~\eqref{eq:SetSelection}, $Q \subseteq \ui$, $\alpha \ge 1$ and $G \subseteq \ui\setminus Q$:
	\begin{enumerate}
		\item If $b_S'(Q) < 1$ for all $S \in \fs$ and $\gs(Q,G) \ge \frac{1}{\alpha} \cdot \max_{I_i \in \ui\setminus Q} \gs(Q,I_i)$, then $G$ satisfies the first condition of~\Cref{def:greedyalpha} and, thus, $\alpha$-approximates the greedy choice. 
		\item If $b'(Q) \ge 1$ and $\gc(Q,G) \ge \frac{1}{\alpha} \cdot \max_{I_i \in \ui\setminus Q} \gc(Q,I_i)$, then $G$ satisfies the second condition of~\Cref{def:greedyalpha} and, thus, $\alpha$-approximates the greedy choice. 
	\end{enumerate}
\end{restatable}

\begin{proof}
	First, assume that $b'_S(Q) < 1$ for all $S \in \fs$ and consider a set $G \subseteq \ui \setminus Q$ with
	$\gs(Q,G) \ge \frac{1}{\alpha} \cdot \max_{I_i \in \ui\setminus Q} \gs(Q,I_i)$.
	
	Let $I^* = \arg\max_{I_i \in \ui\setminus Q} A(Q) - A(Q \cup \{I_i\}) = \arg\max_{I_i \in \ui\setminus Q} \gs(Q,I_i)$.
	By assumption $b'_S(Q) < 1$ for all $S \in \fs$ and, as we consider a scaled instance, $a'_i \ge 1$ for all $I_i \in \ui$.
	Thus, the remaining instance is a set cover instance as adding an interval $I_i$ to the solution satisfies all constraints $S$ with $I_i \in S$ that have not already been satisfied by $Q$.
	
	Let $\OPT_Q$ denote the optimal solution for the remaining instance after adding $Q$ to the solution, i.e., the optimal solution to~\eqref{eq:SetSelection:param}.
	Using a standard set cover argument, we can observe that $\frac{A(Q)}{\OPT_Q} \le A(Q) - A(Q \cup \{I^*\})$ as the optimal solution satisfies the remaining constraints at cost $\OPT_Q$, but a single interval can satisfy at most $A(Q) - A(Q \cup \{I^*\})$ constraints.
	\jnew{Note that this argument only holds because all left-hand side coefficients are at least as large as the right-hand sides. Otherwise, adding an interval $I_i$ later, i.e., after $Q' \supset Q$ has already been added to the solution, could satisfy more constraints, i.e., $A(Q)-A(Q\cup \{I_i\}) < A(Q')-A(Q'\cup \{I_i\})$. This is one of the reasons why the offline greedy algorithm uses two greedy criteria.}
	
	By assumption and definition of $\gs$, we have $\alpha \cdot (A(Q)-A(Q\cup G)) \ge \max_{I_i \in \ui\setminus Q} A(Q)-A(Q \cup\{I_i\})$ and, therefore, $\frac{A(Q)}{\OPT_Q} \le \alpha \cdot (A(Q)-A(Q\cup G))$.
	Rearranging the latter inequality, we obtain 
	$A(Q\cup G) \le A(Q) \cdot \left(1-\frac{1}{\alpha \OPT_Q}\right)$.
	Since $\OPT \ge \OPT_Q$ for the optimal solution $\OPT$ of the complete instance, we get $A(Q\cup G) \le A(Q) \cdot \left(1-\frac{1}{\alpha \OPT}\right)$.
	This implies that $G$ satisfies the first condition of~\Cref{def:greedyalpha}.
	
	Next, assume $b'(Q) \ge 1$ and consider a set $G \subseteq \ui \setminus Q$ with
	$\gc(Q,G) \ge \frac{1}{\alpha} \cdot \max_{I_i \in \ui\setminus Q} \gc(Q,I_i)$.
	
	Let $I^* = \arg\max_{I_i \in \ui\setminus Q} b'(Q) - b'(Q \cup \{I_i\})= \arg\max_{I_i \in \ui\setminus Q} \gc(Q,I_i)$.
	Observe that $\frac{b'(Q)}{\OPT_Q} \le b'(Q) - b'(Q \cup \{I^*\})$  as the optimal solution covers the remaining constraints at cost $\OPT_Q$, but a single interval can decrease the total slack between left-hand and right-hand sides of~\eqref{eq:SetSelection:param}  by at most $b'(Q) - b'(Q \cup \{I^*\})$.
	By assumption and definition of $\gc$, we have $\alpha \cdot (b'(Q)-b'(Q\cup G)) \ge \max_{I_i \in \ui\setminus Q} b'(Q)-b'(Q \cup\{I_i\})$ and, therefore, $\frac{b'(Q)}{\OPT_Q} \le \alpha \cdot (b'(Q)-b'(Q\cup G))$.
	Rearranging the latter inequality, we obtain 
	$b'(Q\cup G) \le b'(Q) \cdot \left(1-\frac{1}{\alpha \OPT_Q}\right)$.
	Since $\OPT \ge \OPT_Q$ for the optimal solution $\OPT$ of the complete instance, we get $b'(Q\cup G) \le b'(Q) \cdot \left(1-\frac{1}{\alpha \OPT}\right)$.
	This and the assumption $b'(Q) \ge 1$ imply that $G$ satisfies the second condition of~\Cref{def:greedyalpha}.
	Note that the argument for the second case does not use that all non-zero coefficients are at least one. Thus, the statement also holds if there are  coefficients $0 < a_i' < 1$.
\end{proof}

With the following lemma, we bound the number of iterations $j$ in which $G_j$ $\alpha$-approximates the current greedy choice via an adjusted set cover greedy analysis.

\begin{restatable}{lemma}{greedyapproxa}
	\label{lem:greedy:approx1}
	Consider an arbitrary algorithm for~\eqref{eq:SetSelection} that scales the coefficients by factor $\gamma$ and iteratively adds disjoint subsets $G_1,\ldots,G_h$ of $\ui$  to the solution until the instance is solved.
	The number of groups $G_j$ that $\alpha$-approximate the current greedy choice (after adding $Q_j = \bigcup_{1 \le j' \le j-1} G_{j'}$) is at most
	$\alpha \cdot \grset(\gamma) \cdot \OPT.$
\end{restatable}

\begin{proof}
	We first show that the number of iterations $j$ with $b'(Q_j) \ge 1$ and $b'(Q_j \cup G_j) \le (1 - \frac{1}{\alpha \cdot \OPT}) \cdot b'(Q_j)$, i.e., the number of iterations that satisfy the second condition of~\Cref{def:greedyalpha}, is at most 
	$\alpha \lceil\ln(\gamma \cdot m\cdot \max_S (w^*-L_{S}))\rceil \cdot \OPT$.
	
	Let $\bar{G}_1,\ldots, \bar{G}_k \subseteq \ui$ denote the sets that are added to the solution by the algorithm \emph{and} satisfy the second condition of~\Cref{def:greedyalpha}.
	Assume that the sets are indexed in the order they are added.
	For each $j \in \{1,\ldots,k\}$, let $\bar{Q}_j \subseteq \ui$ denote the set of intervals that are added to the solution \emph{before} $\bar{G}_j$.
	Note that $\{\bar{G}_1,\ldots,\bar{G}_{j-1}\} \subseteq \bar{Q}_j$, but $\bar{Q}_j$ might contain additional added groups that just do not satisfy the second condition of~\Cref{def:greedyalpha}.
	
	By assumption, $b'(\bar{Q}_j \cup \bar{G}_j) \le (1 - \frac{1}{\alpha \cdot \OPT}) \cdot b'(\bar{Q}_j)$.
	A recursive application of this inequality and the fact that $(1-x) < e^{-x}$ for all $x \in \mathbb{R}\setminus \{0\}$ implies:
	$$
	b'(\bar{Q}_j \cup \bar{G}_j) \le b'(\emptyset) \cdot \left(1-\frac{1}{\alpha \OPT}\right)^j < b'(\emptyset) \cdot e^{-\frac{j}{\alpha \OPT}}.
	$$
	Thus, after $j = \alpha \cdot \OPT \cdot \lceil\ln b'(\emptyset)\rceil$ iterations that satisfy the second condition of~\Cref{def:greedyalpha}, we have $b'(\bar{Q}_j \cup \bar{G}_j) < b'(\emptyset) \cdot e^{- \ln b'(\emptyset)} = 1$.
	But if $b'(\bar{Q}_j \cup \bar{G}_j) < 1$, then there can be no further iteration that satisfies the second condition of~\Cref{def:greedyalpha}.
	Thus, the number of such iterations is at most $\alpha \cdot \OPT \cdot \lceil\ln b'(\emptyset)\rceil$. 
	Since $b'(\emptyset)$ is upper bounded by $\gamma \cdot m\cdot \max_S (w^*-L_{S})$ as we have $m$ constraints with scaled right-hand side values of at most $\gamma \cdot \max_S (w^*-L_{S})$, 
	the number of such iterations is at most $\alpha \cdot \lceil\ln(\gamma \cdot m\cdot \max_S (w^*-L_{S})\rceil \cdot \OPT$.
	
	Next, we show that the number of iterations $j$ with $A(Q_j \cup G_j) \le (1 - \frac{1}{\alpha \cdot \OPT}) \cdot A(Q_j)$, i.e., the number of iterations that satisfy the first condition of~\Cref{def:greedyalpha}, is at most 
	$\alpha \lceil\ln(m)\rceil \cdot \OPT$.
	
	Let $\bar{G}_1,\ldots, \bar{G}_k \subseteq \ui$ denote the sets that are added to the solution by the algorithm and satisfy the first condition of~\Cref{def:greedyalpha}.
	Assume that the sets are indexed in the order they are added.
	For each $j \in \{1,\ldots,k\}$, let $\bar{Q}_j \subseteq \ui$ again denote the set of intervals that are added to the solution \emph{before} $\bar{G}_j$.
	Note that $\{\bar{G}_1,\ldots,\bar{G}_{j-1}\} \subseteq \bar{Q}_j$, but $\bar{Q}_j$ might contain additional sets that just do not satisfy the first condition of~\Cref{def:greedyalpha}.
	
	By assumption, $A(\bar{Q}_j \cup \bar{G}_j) \le (1 - \frac{1}{\alpha \cdot \OPT}) \cdot A(\bar{Q}_j)$.
	A recursive application of this inequality and the fact that $(1-x) < e^{-x}$ for all $x \in \mathbb{R}\setminus \{0\}$ implies:
	$$
	A(\bar{Q}_j \cup \bar{G}_j) \le A(\emptyset) \cdot \left(1-\frac{1}{\alpha \OPT}\right)^j < A(\emptyset) \cdot e^{-\frac{j}{\alpha \OPT}}.
	$$
	Thus, after $j = \alpha \cdot \OPT \cdot \lceil\ln A(\emptyset)\rceil$ such iterations, we have $A(\bar{Q}_j \cup \bar{G}_j) < A(\emptyset) \cdot e^{- \ln A(\emptyset)} = 1$.
	But if $A(\bar{Q}_j \cup \bar{G}_j) < 1$, then $A(\bar{Q}_j \cup \bar{G}_j) = 0$ and the instance is solved and no further iteration is executed.
	Since $A(\emptyset)$ is upper bounded by the number of constraints $m$, the number of iterations that satisfy the first condition of~\Cref{def:greedyalpha} is at most $\alpha \lceil\ln(m)\rceil \cdot \OPT$.
	
	In total, at most $\alpha \lceil\ln(m)\rceil \cdot \OPT$ iterations satisfy the first condition of~\Cref{def:greedyalpha} and at most
	$\alpha \cdot \lceil\ln(\gamma \cdot m\cdot \max_S (w^*-L_{S})\rceil \cdot \OPT$ iterations satisfy the second condition of~\Cref{def:greedyalpha}. 
	In summation, there are at most $\alpha \cdot (\lceil\ln(\gamma \cdot n\cdot \max_{e \in E} b_e)\rceil + \lceil\ln(n)\rceil) \cdot \OPT = \alpha \cdot \grset(\gamma) \cdot \OPT$ iterations that satisfy~\Cref{def:greedyalpha}.
\end{proof}

The lemma states that the number of groups $G_j$ that $\alpha$-approximate their greedy choice is within a factor of $\alpha$ of the performance guarantee $\grset(\gamma)$ of the offline greedy algorithm.
If each $G_j$ $\alpha$-approximates its greedy choice, the iterative algorithm achieves an approximation factor of  $\max_j |G_j| \cdot \alpha \cdot \grset(\gamma)$.
Thus, approximating the greedy choices by a constant factor using a constant group size is sufficient to only lose a constant factor compared to the offline greedy algorithm.

This insight gives us a framework to solve \minset under uncertainty.
Recall that the $w_i$'s (and by extension the $a_i$'s and $b_S$'s) are uncertain and only revealed once we irrevocably add an $I_i \in \ui$ to the solution. 
We refer to a revealed $w_i$ as a \emph{query result}, and to a fixed set of revealed $w_i$'s for all $I_i \in \ui$ as a \emph{realization of query results}.
Consider an iterative algorithm.
The sets $G_j$ can be computed and queried adaptively and are allowed to depend on (random) query results from previous iterations. Hence, 
$X_j = |G_j|$ is a random variable. 
Let $Y_j$ be an indicator variable denoting whether the algorithm executes iteration $j$ ($Y_j = 1$) or terminates beforehand ($Y_j = 0$). 
We define the following class of iterative algorithms and show that algorithms from this class achieve certain guarantees.

\begin{definition}
	\label{def:alphagreedy}
	An iterative algorithm is \alphagreedy if it satisfies:
	\begin{enumerate}
		\item For every realization of query results; each $G_j$ $\alpha$-approximates the greedy choice as characterized by $Q_j$ on the instance with coefficients scaled by~$\gamma$.
		\item $\EX[X_j \mid Y_j = 1] \le \beta$ holds for all iterations $j$.
	\end{enumerate} 
\end{definition}


\begin{theorem}
	\label{theo:framework:minset}
	Each \alphagreedy algorithm for \minset under uncertainty achieves a competitive ratio of $\alpha \cdot \beta \cdot \grsetu(\gamma) \in \mathcal{O}(\alpha\cdot\beta\cdot \log(m))$.
\end{theorem}

\begin{proof}
	Consider an \alphagreedy algorithm $\alg$ for \minset.
	The expected cost of $\ALG$ is $\EX[\ALG] = \sum_j \EX[X_j]$.
	Using the total law of expectations, we get
	\begin{align*}
		\EX[\ALG] &= \sum_j \PR[Y_j = 1] \EX[X_j \mid Y_j = 1] + \PR[Y_j = 0] \EX[X_j \mid Y_j = 0]\\ &= \sum_j \PR[Y_j = 1] \EX[X_j \mid Y_j = 1],
	\end{align*}
	where the last inequality holds as $\EX[X_j \mid Y_j = 0] = 0$ by definition (if the algorithm terminates before iteration $j$, then it adds no more elements to the solution and, thus, $X_j = 0$).	
	By the second property of~\Cref{def:alphagreedy}, this implies $\EX[\ALG] \le \beta \cdot \sum_j \PR[Y_j = 1]$.
	
	Thus, it remains to bound $\sum_j \PR[Y_j = 1]$, which corresponds to the expected number of iterations of $\ALG$. 
	Consider a fixed realization of query results, then, by the first property of \textsc{($\alpha,\beta,\gamma$)-Greedy}, each $G_j$ $\alpha$-approximates its greedy choice for the~\eqref{eq:SetSelection} instance of the realization scaled by factor $\gamma$.
	Then, \Cref{lem:greedy:approx1} implies that the number of iterations is at most $\alpha \cdot \grset(\gamma)\cdot \OPT$, which is upper bounded by $\alpha \cdot \grsetu(\gamma) \cdot \OPT$.
	As this upper bound on the number of iterations holds for every realization and $\OPT$ is the only random variable of that term (since we substituted $\grset$ with $\grsetu$), we can conclude $\sum_j \PR[Y_j = 1] \le \alpha \cdot \grsetu(\gamma) \cdot \EX[\OPT]$, which implies $\EX[\ALG] \le \alpha \cdot \beta \cdot \grsetu(\gamma) \cdot \EX[\OPT]$.
\end{proof}

\section{\minset with Deterministic Right-Hand Sides and More Covering Problems}
\label{sec:mincover}

We first introduce an algorithm for a variant of \minset under uncertainty, where the right-hand sides $b_S$ of the ILP representation~\eqref{eq:SetSelection} are deterministic and explicit part of the input.
Afterwards, we generalize the algorithm for more general covering problems and a different balancing parameter.

\subsection{\minset With Deterministic Right-Hand Sides}

We consider a variant of \minset under uncertainty, where the right-hand sides $b_S$ of the ILP representation~\eqref{eq:SetSelection} are deterministic and explicit part of the input.
Thus, only the coefficients $a_i = (w_i - L_i)$ remain uncertain within the interval $(0,U_i - L_i)$.
For this problem variant, it can happen that the instance has no feasible solution. In that case, we require every algorithm (including $\OPT$) to reduce the covering requirements as much as possible.
As we consider the stochastic problem variant, recall that the balancing parameter is defined as $\tau = \min_{I_i \in \ui} \tau_i$ for $\tau_i = \PR[w_i \ge \frac{U_i+L_i}{2}]$. 



\begin{algorithm}[tb]
	\KwIn{Instance of \minset with deterministic right-hand sides.}
	$Q = \emptyset$;
	Scale $a$ and $b$ by $\frac{2}{s_{\min}}$ to $a'$ and $b'$ for $s_{\min} = \min_{I_i \in \ui \colon U_i - L_i > 0} U_i - L_i$\;
	\While{the problem is not solved}{
		\leIf{$b'(Q) \ge 1$}{$g = \ogc$}{$g=\ogs$}\label{algcov:greedydef1}
		\Repeat{the problem is solved or $w_i -L_i \ge \frac{1}{2} \cdot (U_i-L_i)$\label{algocov:cancel} 
		}{
			$I_i \gets \arg\max_{I_j \in \ui \setminus Q} g(Q, I_j)$; Query $I_i$; $Q \gets Q \cup \{I_i\}$\;
		}	
	}
	\caption{\minset with deterministic right-hand sides.}
	\label{alg:covering1}
\end{algorithm}

\begin{theorem}
	\label{theo:covering}
	There is an algorithm for \minset under uncertainty with deterministic right-hand sides and a competitive ratio of $\frac{2}{\tau} \cdot \grset(\gamma) \in \mathcal{O}(\frac{1}{\tau} \cdot \log  m)$ with $\gamma = 2/s_{\min}$ for $s_{\min} = \min_{I_i \in \ui \colon U_i - L_i > 0} U_i - L_i$.
\end{theorem}

The algorithm of the theorem loses only a factor $\frac{2}{\tau}$ compared to the greedy approximation factor $\grset(\gamma)$ on the corresponding offline problem.
We show the theorem by proving that Algorithm~\ref{alg:covering1} is an \alphagreedy algorithm for $\alpha = 2$, $\beta = \frac{1}{\tau}$ and $\gamma = \frac{2}{s_{\min}}$ with $s_{\min} = \min_{I_i \in \ui \colon U_i - L_i > 0} U_i - L_i$.
Then, \Cref{theo:framework:minset} implies the theorem.
We remark that we scale by $\frac{2}{s_{\min}}$ instead of $\frac{1}{s_{\min}}$ because of technical reasons that will become clear in the proof of the theorem.

The algorithm scales the coefficients by factor $\gamma$; we use $a'$ and $b'$ to refer to the scaled coefficients.
The idea of Algorithm~\ref{alg:covering1} is to execute the greedy Algorithm~\ref{alg_verification_covering_1} under the assumption that $a_i = U_i-L_i$ (and $a'_i = \gamma a_i$) for all $I_i \in \ui$ that were not yet added to the solution.
As $a_i = (w_i-L_i) \in (0,U_i-L_i)$, this means that we assume $a_i$ to have the largest possible value.
Consequently, $s_{\min}$ is the smallest (non-zero) coefficient $a_i$ under this assumption.
The algorithm computes the greedy choice based on the \emph{optimistic greedy values} $$\ogc(Q,I_i) = {\sum_{S \in \fs \colon I_i \in S} b'_S(Q) - \max\{0, b'_S(Q)-\gamma(U_i-L_i)\}}$$ (if $b'(Q) \ge 1$) and $$\ogs(Q,I_i) = {|\{S\in \fs \colon I_i \in S \mid  b'_S(Q) > 0 \land b'_S(Q) - \gamma(U_i-L_i) \le 0\}|}$$ (otherwise).
That is, the greedy values under the assumption $a_i = U_i-L_i$.
We call these values optimistic as they might overestimate but never underestimate the actual greedy values. 
For subsets $G \subseteq \ui$, we define $\ogs(Q,G)$ and $\ogc(Q,G)$ analogously.

In contrast to $\gs$ and $\gc$, Algorithm~\ref{alg:covering1} has sufficient information to compute $\ogs$ and $\ogc$, and, therefore, the best greedy choice based on the optimistic greedy values.
The algorithm is designed to find, in each iteration, an element $I_i$ with $\gc(Q,I_i) \ge \frac{1}{2} \cdot \ogc(Q,I_i)$ for the current $Q$ (or analogously for $\ogs$ and $\gs$).
We show that (i) this ensures that each iteration $2$-approximates the greedy choice and (ii) that finding such an element takes only $\frac{1}{\tau}$ tries in expectation.

\begin{proof}[Proof of \Cref{theo:covering}]
	Let $j$ be an arbitrary iteration of the outer while-loop, $X_j$ denote the number of queries during the iteration, and $Y_j$ indicate whether the algorithm executes iteration $j$ ($Y_j = 1$) or not ($Y_j = 0$).
	
	Assuming $Y_j = 1$, the algorithm during iteration $j$ executes queries to elements $I_i$ until either $w_i-L_i \ge \frac{1}{2} (U_i-L_i)$ or the problem is solved. 
	Since $w_i \ge \frac{(U_i+L_i)}{2}$ implies $w_i-L_i \ge \frac{1}{2} (U_i-L_i)$ and $\PR[w_i \ge \frac{(U_i+L_i)}{2}] \ge \tau$ holds by assumption,
	the number of attempts until the current $I_i$ satisfies the inequality follows a geometric distribution with success probability at least $\tau$. Hence, $\EX[X_j \mid Y_j = 1] \le \frac{1}{\tau}$; proving Property $2$~of~\Cref{def:alphagreedy}.
	
	We continue by proving Property $1$~of~\Cref{def:alphagreedy}.
	Consider a fixed realization. 
	Let $\bar{G}_j$ denote the queries of iteration $j$ \emph{except} the last one and let $I_{\bar{j}}$ denote the last query of iteration $j$.
	Then $G_j = \bar{G}_j \cup \{I_{\bar{j}}\}$ is the set of queries during the iteration.
	Finally, let $Q_j$ denote the set of queries before iteration $j$.
	We show that $G_j$ $2$-approximates the greedy choice of the scaled instance, which implies  Property~$1$ of~\Cref{def:alphagreedy}. 
	
	If the iteration solves the problem, then $G_j$ clearly $1$-approximates the greedy choice and we are done. 
	Thus, assume otherwise.
	We distinguish between the two cases (1) $b'(Q_j) \ge 1$ and (2) $b'(Q_j) < 1$.
	
	\textbf{Case (1):} We show first that $G_j$ $2$-approximates the greedy choice if $b'(Q_j) \ge 1$.
	In this case, we have $g = \ogc$ (cf.~Line~\ref{algcov:greedydef1}).
	By choice of $I_{\bar{j}}$, we have $\ogc(Q_j \cup \bar{G}_j,I_{\bar{j}}) = \max_{I_i \in \ui \setminus (Q_j \cup \bar{G}_j)} \ogc(Q_j \cup \bar{G}_j,I_i)$, i.e., $I_{\bar{j}}$ has the best optimistic greedy value when it is chosen.
	
	As the iteration does not solve the instance, we have $(w_{\bar{j}}-L_{\bar{j}}) \ge \frac{1}{2}(U_{\bar{j}}-L_{\bar{j}})$ by Line~\ref{algocov:cancel}.
	This directly implies that the real greedy value of $I_{\bar{j}}$ is at least half the optimistic greedy value, i.e., $\gc(Q_j \cup \bar{G}_j,I_{\bar{j}}) \ge \frac{1}{2} \cdot \ogc(Q_j \cup \bar{G}_j,I_{\bar{j}})$. 
	
	Since the best optimistic greedy value is never smaller than the best real greedy value, we get $\gc(Q_j \cup \bar{G}_j,I_{\bar{j}}) \ge \frac{1}{2}\cdot \max_{I_i \in \ui \setminus (Q_j \cup \bar{G}_j)} \gc(Q_j \cup \bar{G}_j,I_i)$.
	This allows us to apply~\Cref{lem:greedy:alpha} to get  
	$b'(Q_j \cup \bar{G}_j \cup \{I_{\bar{j}}\}) \le (1- \frac{1}{2\OPT}) \cdot b'(Q_j \cup \bar{G}_j)$.
	Using  $b'(Q_j \cup \bar{G}_j)\le b'(Q_j)$ and $G_j = \bar{G}_j \cup \{I_{\bar{j}}\}$, we can conclude $b'(Q_j \cup G_j) \le (1- \frac{1}{2\OPT}) \cdot b'(Q_j)$, which shows that $G_j$ satisfies Condition~$2$ of~\Cref{def:greedyalpha}.

	\textbf{Case (2):} Next, we show that $G_j$ $1$-approximates the greedy choice if $b'(Q_j) < 1$.
	In this case, we have $g = \ogs$ (cf.~Line~\ref{algcov:greedydef1}).
	Similar to the previous case, we have $\ogs(Q_j \cup \bar{G}_j,I_{\bar{j}}) = \max_{I_i \in \ui \setminus (Q_j \cup \bar{G}_j)} \ogs(Q_j \cup \bar{G}_j,I_i)$, i.e., $I_{\bar{j}}$ has the best optimistic greedy value when it is chosen.
	
	From $b'(Q_j) = \sum_{S \in \fs} b'_S(Q_j) < 1$ follows $b_S'(Q_j) < 1$ for all $S \in \fs$.
	Furthermore, every element $I_i$ with $w_i - L_i \ge \frac{1}{2} U_i-L_i$ satisfies $a_i = w_i - L_i \ge \frac{s_{\min}}{2}$ and, therefore $a_i' = \gamma a_i = \frac{2}{s_{\min}} \cdot a_i \ge 1$.
	This means that adding $I_i$ to the solution satisfies \emph{all} constraints for sets $S$ with $I_i \in S$ that are previously not satisfied.
	Thus, the optimistic greedy value $\ogs(Q_j,I_i)$ and the real greedy value $\gs(Q_j,I_i)$ are the same, i.e., $\ogs(Q_j,I_i)= \gs(Q_j,I_i)$, as adding $I_i$ cannot satisfy more constraints even if the coefficient $a_i$ was $U_i-L_i$.
	This observation is crucial for the remaining proof and the reason we scale with $\gamma = \frac{2}{s_{\min}}$ instead of $\frac{1}{s_{\min}}$.
	
	Since the iteration does not solve the instance by assumption, we have $(w_{\bar{j}}-L_{\bar{j}}) \ge \frac{1}{2}(U_{\bar{j}}-L_{\bar{j}})$ by Line~\ref{algocov:cancel}.
	As argued above, this implies that $I_{\bar{j}}$ has the best optimistic \emph{and} actual greedy value when it is added to the solution, i.e., $\gs(Q_j \cup \bar{G}_j,I_{\bar{j}}) = \ogs(Q_j \cup \bar{G}_j,I_{\bar{j}}) = \max_{I_i \in \ui \setminus (Q_j \cup\bar{G}_j)} \ogs(Q_j \cup \bar{G}_j,I_i)$. 
	Thus, even under the assumption that all elements $I_i$ of $\mathcal{I}\setminus(Q_j \cup \bar{G}_j)$ have coefficients $a_i = (U_i-L_i)$, interval $I_{\bar{j}}$ achieves the best greedy value.
	
	Let $LB$ denote the optimal solution value for the remaining instance after querying $Q_j \cup \bar{G}_j$ under exactly this assumption that $a_i = (U_i-L_i)$ and $a_i' = \gamma(U_i-L_i)$ for all $I_i \in \mathcal{I}\setminus(Q_j \cup \bar{G}_j)$.
	Clearly $LB \le \OPT$.
	
	Under the assumption that $a'_i = \gamma(U_i-L_i)$ for all $I_i \in \mathcal{I}\setminus(Q_j \cup \bar{G}_j)$, the instance is scaled (i.e., it satisfies that all non-zero coefficients are at least one).
	Thus, we can apply~\Cref{lem:greedy:alpha} under the assumption
	to get
	$A(Q_j \cup \bar{G}_j \cup \{I_{\bar{j}}\}) \le (1- \frac{1}{LB}) \cdot A(Q_j \cup \bar{G}_j)$.
	Since $LB \le \OPT$, this gives us $A(Q_j \cup \bar{G}_j \cup \{I_{\bar{j}}\}) \le (1- \frac{1}{\OPT}) \cdot A(Q_j \cup \bar{G}_j)$.
	Using  $A(Q_j \cup \bar{G}_j)\le A(Q_j)$ and $G_j = \bar{G}_j \cup \{I_{\bar{j}}\}$, we can conclude $A(Q_j \cup G_j) \le (1- \frac{1}{\OPT}) \cdot A(Q_j)$, which shows that $G_j$ satisfies Condition~$2$ of~\Cref{def:greedyalpha} for $\alpha = 1$.
\end{proof}

\subsection{Multiset Multicover under Stochastic Explorable Uncertainty}

We extend the result of the previous section to a stochastic variant of the \emph{multiset multicover problem} (see, e.g.,~\cite{rajagopalan1998}), which generalizes the classical set cover problem.
We are given a universe $\mmu$ of $n$ elements and a family $\mms$ of $m$ multi-sets with $M \subseteq \mmu$ for all $M \in \mms$.
For each $e \in \mmu$, we are given a deterministic covering requirement $b_e \in \mathbb{R}_+$.
Each multi-set $M \in \mms$ contains a number of (fractional) copies of element $e \in \mmu$, denoted by $a_{M,e} \in \mathbb{R}_+$. 
The goal is to select a subset $C \subseteq \mms$ of minimum cardinality that satisfies all covering requirements:
\begin{equation}
	\label{eq:multiset}
	\tag{{\sc MinCoverIP}}
	\begin{array}{llll}
		\min &\sum_{M \in \mms} x_M\\
		\text{s.t. }& \sum_{M \in \mms} x_{M} \cdot a_{M,e} &\ge b_e &\forall e \in \mmu\\
		& \hfill x_M & \in \{0,1\}& \forall M \in \mms.
	\end{array}
\end{equation}
In our stochastic variant,
\mincover under uncertainty, the constraints are uncertain; more precisely, the coefficients $a_{M,e}$ are initially unknown. We are given uncertainty intervals $I_{M,e} = (L_{M,e},U_{M,e})$ with $a_{M,e} \in I_{M,e}$. We may query a set $M \in \mms$, which reveals the precise values $a_{M,e}$ for set $M$ and \emph{all} $e \in \mmu$. In contrast to \minset, we study \mincover in a \emph{query-commit model}, which means that we can only add a set $M \in \mms$ to the solution if we query it, and we \emph{have} to add all queried sets to the solution.  
In a sense, we solve~\eqref{eq:multiset} with uncertainty in the coefficients $a_{M,e}$ and irrevocable decisions:
Once we add a set $M$ to the solution ($x_M = 1$), all $a_{M,e}$ are revealed and we can never remove $M$ from the solution.
Irrevocably adding a set corresponds to querying it.
The goal is to find a feasible solution to the ILP with a minimal number of queries. 

From the problem description alone, the problem of solving \mincover under uncertainty is quite similar to solving~\eqref{eq:SetSelection} under uncertainty. 
The difference is that in~\eqref{eq:SetSelection}, each interval $I_i$ has the same coefficient $a_i$ in all constraints where the variable occurs, i.e., where the variable has a non-zero coefficient. This value is drawn from a single distribution over $I_i$.
For \mincover under uncertainty on the other hand, each coefficient $a_{M,e}$ of a multiset $M$ can be different and has its own uncertainty interval $I_{M,e}$.
In contrast to \minset under uncertainty, we still have deterministic right-hand sides.

\subsubsection{Offline \mincover.} All our results for offline~\minset directly translate to offline \mincover, where we assume full knowledge of all coefficients $a_{M,e}$.
To see this, we again define an ILP parametrized by the set $Q \subseteq \mms$ of multisets that have already been added to the solution and adjust the right-hand sides to the remaining covering requirement after adding $Q$.
Here, $b_e(Q) =  \max\{b_e - \sum_{M \in Q} a_{M,e} ,0\}$ and $b(Q) = \sum_{e \in \mmu} b_e(Q)$.
\begin{equation}
	\tag{{\sc MinCoverIP-Q}}
	\label{eq:multiset:param}
	\begin{array}{llll}
		\min &\sum_{M \in \mms\setminus Q} x_M\\
		\text{s.t. }& \sum_{M \in \mms\setminus Q} x_M \cdot a_{M,e} &\ge b_e(Q) &\forall e \in \mmu\\
		& x_M & \in \{0,1\}& \forall M \in \mms \setminus Q\\
	\end{array}
\end{equation}
Based on this ILP, we define an offline greedy algorithm for \mincover in the same way as for offline \minset, only with slightly different greedy value definitions.
The algorithm scales the coefficients such that all non-zero left-hand side coefficients are at least $1$, and then greedily adds the set to the solution that reduces the right-hand sides the most, i.e., the set $M \in \mms$ that maximizes $\gc(Q,M) = b'(Q) - b'(Q\cup\{M\})$ ($a'$ and $b'$ indicate scaled coefficients). 
For a set $G \subseteq \mms$, we define $\gc(Q,G) = b'(Q) - b'(Q \cup G)$.

After $b_e'(Q) < 1$ for all $e \in \mmu$, we can exploit that all scaled non-zero coefficients $a'_{M,e}$ are at least one. 
This means that adding a set $M$ satisfies all remaining constraints of elements $e$ with $a'_{M,e} > 0$.
Thus, the remaining problem reduces to a \setcover instance, which can be solved by using the greedy algorithm by Chvatal~\cite{chvatal1979}, which greedily adds the set $M$ that maximizes $\gs(Q,M) = A(Q) - A(Q \cup \{M\})$ with $A(Q) = |\{e \in \mmu \mid b_e(Q) > 0\}|$, i.e., the set that satisfies the largest number of constraints that are not already satisfied by $Q$.
For a set $G \subseteq \mms$, we define $\gs(Q,G) =  A(Q) - A(Q \cup G)$.
We call $\gc(Q,M)$ and $\gs(Q,M)$  the \emph{greedy values} of $M$ after adding $Q$.

This offline algorithm achieves an approximation factor of $\rho'(\gamma) = \lceil\ln(\gamma \cdot n \cdot \max_{e \in \mmu} b_e)\rceil + \lceil\ln(n)\rceil$ with $s_{\min} = \min_{e \in \mmu, M \in \mms \colon a_{M,e} > 0} a_{M,e}$, $\gamma = 1/s_{\min}$ and $n = |\mmu|$ (follows from~\cite{dobson1982}).
It is easy to see that~\Cref{def:greedyalpha},~\Cref{lem:greedy:alpha},~\Cref{lem:greedy:approx1},~\Cref{def:alphagreedy} and~\Cref{theo:framework:minset} all transfer to \mincover using $\rho'$ instead of $\rho$ (or $\grsetu$) and using the adjusted greedy value definitions. 
This is because none of the proofs actually uses the special case properties of~\eqref{eq:SetSelection} that~\eqref{eq:multiset} does not have.
In the following, we will use these definitions, lemmas, and theorems to prove algorithmic results for stochastic \mincover.

\subsubsection{\mincover under uncertainty.} Similar to \minset with deterministic right-hand sides, it can happen that the instance has no feasible solution.
In that case, we require every algorithm (including $\OPT$) to reduce the covering requirements as much as possible.

\jnew{Our algorithm again relies on the optimistic greedy values $\ogc(Q,M) = b'(Q)- \sum_{e \in \mmu} \max\{0,b'_e(Q) - \gamma U_{M,e}\}$ (if $b'(Q) \ge 1$) and $\ogs(Q,M) = A(Q) - |\{e \in \mmu \mid b_e'(Q) - \gamma U_{M,e} > 0\}|$ (otherwise).
	For sets $G \subseteq \mms$, we define $\ogs(Q,G)$ and $\ogc(Q,G)$ analogously.
	Furthermore, the algorithm again scales the left-hand side coefficients by 
	$\gamma = 2/s_{\min}$ for $s_{\min} = \min_{e \in \mmu, M \in \mms \colon U_{M,e} > 0} U_{M,e}$.
}

In contrast to \minset under uncertainty, we need a different balancing parameter to take the more general nature of the coefficients into account.
\jnew{
	For \minset under uncertainty\todo{There was a mistake with the original balancing parameter, so I had to change quite a lot here. (Non of this was part of the IPCO paper). If it is too artificial now, we can think about removing 4.2 completely}, the special structure of the coefficients for an interval $I_i$ allowed us to relate the event $w_i \ge \frac{U_i+L_i}{2}$ to the actual and optimistic greedy value of $I_i$ after some set $Q$ has already been queried. For \mincover the greedy values of a multiset $M$ depend on the outcomes of several random variables, one for each $e \in E$ with $U_{M,e} > 0$. In particular, the greedy value $\gs(Q,M)$ now depends on the number of random variables $a_{M,e}$ with a sufficiently large value that satisfies the corresponding constraint. Similar, the greedy value $\gc(Q,M)$ now depends on a sum of random variables. Even if $a_{M,e} \ge \frac{U_{M,e}+L_{M,e}}{2}$ holds with probability $\tau$ for each $M$ and $e$, this property does not necessarily translate to the greedy values.
	
	For these reasons, we define the balancing parameter $\tau'$ based on the greedy values for a scaled instance.
	We characterize the distributions for an $M \in \mms$ by $\tau'_M = \min\{\tau'_{M,1} , \tau'_{M,2} \}$,  where $\tau'_{M,1}$ is the minimum probability $\Pr[\gc(M,Q) \ge \frac{1}{2} \ogc(M,Q)]$ over all $Q \subseteq \mms\setminus\{M\}$ \emph{and} all realizations of the elements of $Q$.
	Similar, $\tau'_{M,2}$ is the minimum probability $\Pr[\gs(M,Q) \ge \frac{1}{2} \ogs(M,Q)]$ over all $Q \subseteq \mms\setminus\{M\}$ \emph{and} all realizations of the elements of $Q$ with $b_e'(Q) \le 1$ for all $e \in \mmu$.
	For the complete instance, we use parameter $\tau' = \min_{M \in \mms} \tau'_M$.
	
	In contrast to the parameter for \minset, the parameter $\tau'$ is defined in a more artificial and restrictive way.
	Nevertheless, we argue that small values of $\frac{1}{\tau'}$ still capture interesting distributions.
	For example, if the distribution of each $a_{M,e}$ is symmetric over the center of $I_{M,e}$, then $\tau'=0.5$.
}
Our main result is the following theorem.

\begin{algorithm}[tb]
	\KwIn{Instance if \mincover under stochastic explorable uncertainty.}
	$Q \gets \emptyset$;
	Scale coefficients by $\gamma = (\frac{2}{s_{\min}})$ to $a'$ and $b'$ with $s_{\min} = \min_{e \in \mmu, M \in \mms \colon U_{M,e} > 0} U_{M,e}$\;
	\While{the problem is not solved}{
		\leIf{$b'(Q) \ge 1$}{$g = \ogc$}{$g=\ogs$}\label{algcov:greedydef1c}
		\Repeat{the problem is solved \jnew{or $\gc(Q,M) \ge \frac{1}{2} \ogc(Q,M)$ (if $g= \ogc$) or $\gs(Q,M) \ge \frac{1}{2} \ogs(Q,M)$ (if $g = \ogs$)}}{
			$M \gets \arg\max_{M' \in \mms \setminus Q} g(Q, M')$; Query $M$; $Q \gets Q \cup \{M\}$\;
		}	
	}
	\caption{Algorithm \mincover under stochastic explorable uncertainty.}
	\label{alg:covering1c}
\end{algorithm}

\begin{theorem}
	\label{theo:coveringc}
	There exists an algorithm for \mincover under uncertainty with a competitive ratio of $\frac{2}{\tau'} \cdot \rho'(\gamma) \in \mathcal{O}(\frac{1}{\tau'} \cdot \log  n)$ with $\gamma = 2/s_{\min}$ for $s_{\min} = \min_{e \in \mmu, M \in \mms \colon U_{M,e} > 0} U_{M,e}$.
\end{theorem}

Thus, the algorithm of the theorem loses only a factor $\frac{2}{\tau'}$ compared to the greedy approximation factor $\rho'(\gamma)$ on the corresponding offline problem.
We show the theorem by proving that Algorithm~\ref{alg:covering1c} is an \alphagreedy algorithm for $\alpha = 2$, $\beta = \frac{1}{\tau'}$ and $\gamma = \frac{2}{s_{\min}}$.
Then, \Cref{theo:framework:minset} implies the theorem.
The algorithm scales the coefficients by factor $\gamma$; we use $a'$ and $b'$ to refer to the scaled coefficients.
Algorithm~\ref{alg:covering1c} uses the exact same idea as Algorithm~\ref{alg:covering1} for \minset with deterministic right-hand sides. That is, it executes the greedy algorithm by Dobson~\cite{dobson1982} under the assumption that $a_{M,e} = U_{M,e}$ (and $a'_{M,e} = \gamma U_{M,e}$) for all $e \in \mmu$ and all not yet queried $M \in \mms$.

In contrast to $\gs$ and $\gc$, Algorithm~\ref{alg:covering1c} has sufficient information to compute the optimistic greedy values $\ogs$ and $\ogc$ and, therefore, can compute the best greedy choice based on the optimistic greedy values.
The algorithm is designed to find, in each iteration, one set $M$ with $\gc(Q,M) \ge \frac{1}{2} \cdot \ogc(Q,M)$ for the current $Q$ (or analogously for $\ogs$ and $\gs$).
We show that this ensures that each iteration $2$-approximates the greedy choice for the scaled instance.

\begin{proof}[Proof of \Cref{theo:coveringc}]
	Let $i$ be an arbitrary iteration of the outer while-loop, $X_i$ denote the number of queries during the iteration, and $Y_i$ indicate whether the algorithm executes iteration $i$ ($Y_i = 1$) or not.
	Assuming $Y_i = 1$, the algorithm, during $i$, executes queries to sets $M$ \jnew{until either $\gs(Q,M) \ge \frac{1}{2} \ogs(Q,M)$ (if $g=\ogs$) or $\gc(Q,M) \ge \frac{1}{2} \ogc(Q,M)$  (if $g = \ogc$) or the problem is solved. 
		Thus, the iteration terminates at the latest when it finds an $M$ with $\gs(Q,M) \ge \frac{1}{2} \ogs(Q,M)$ or $\gc(Q,M) \ge \frac{1}{2} \ogc(Q,M)$ depending on $g$. By definition of $\tau'$,}
	the number of attempts until this happens follows a geometric distribution with success probability at least $\tau'$. Hence, $\EX[X_i \mid Y_i = 1] \le \frac{1}{\tau'}$; proving Property $2$~of~\Cref{def:alphagreedy}.
	
	We continue by proving the first property of \alphagreedy.
	Consider a fixed realization. 
	Let $\bar{G}_i$ denote the queries of iteration $i$ \emph{except} the last one and let $\bar{M}$ denote the last query of iteration $i$.
	Then $G_i = \bar{G}_i \cup \{\bar{M}\}$ is the set of queries during the iteration.
	Finally, let $Q_i$ denote the set of queries before iteration $i$.
	We show that $G_i$ $2$-approximates the greedy choice of the scaled instance, which implies  Property~$1$ of~\Cref{def:alphagreedy}. 
	If the iteration solves the problem, then $G_i$ clearly $1$-approximates the greedy choice and we are done. 
	Thus, assume otherwise.
	
	We first show that $G_i$ $2$-approximates the greedy choice 
	if $b'(Q_i) \ge 1$.
	If $b'(Q_i) \ge 1$, we have $g = \ogc$ (cf.~Line~\ref{algcov:greedydef1c}).
	By choice of $\bar{M}$, we have $\ogc(Q_i \cup \bar{G}_i,\bar{M}) = \max_{M \in \mms \setminus (Q_i \cup \bar{G_i})} \ogc(Q_i \cup \bar{G}_i,M)$.
	\jnew{Then, $\gc(Q_i \cup \bar{G}_i,\bar{M}) \ge \frac{1}{2} \cdot \ogc(Q_i \cup \bar{G}_i,\bar{M})$  holds by definition of the algorithm as $g = \ogc$ and iteration $i$ does not solve the instance. This implies 
		$\gc(Q_i \cup \bar{G}_i,\bar{M}) \ge \frac{1}{2}\cdot \max_{M \in \mms \setminus (Q_i \cup \bar{G}_i)} \gc(Q_i \cup \bar{G}_i,M)$.}
	By~\Cref{lem:greedy:alpha}, 
	$b'(Q_i \cup \bar{G_i} \cup \{\bar{M}\}) \le (1- \frac{1}{2\OPT}) \cdot b'(Q_i \cup \bar{G_i})$ and, thus, 
	$b'(Q_i \cup G_i) \le (1- \frac{1}{2\OPT}) \cdot b'(Q_i)$, which implies the property.	
	
	Next, we show that $G_i$ $2$-approximates the greedy choice if $b'(Q_i) < 1$.
	By assumption, we have $b'_e(Q_i) < 1$ for all $e \in \mmu$.
	Recall that we scale by factor $\gamma = \frac{2}{s_{\min}}$ with $s_{\min} = \min_{e \in \mmu, M \in \mms \colon U_{M,e} > 0} U_{M,e}$.
	If a set $M$ satisfies $a_{M,e} \ge \frac{1}{2} \cdot U_{M,e}$ for some $e \in \mmu$ with $U_{M,e} > 0$, then $a'_{M,e} = \gamma a_{M,e} \ge \frac{2}{s_{\min}} \cdot \frac{1}{2} \cdot U_{M,e} \ge 1$.
	Thus, adding $M$ to the solution satisfies the constraint of element $e$ (if it was not already satisfied).
	
	By choice of $\bar{M}$, we have $\ogs(Q_i \cup \bar{G}_i,\bar{M}) = \max_{M \in \mms \setminus (Q \cup \bar{G}_i)} \ogs(Q_i \cup \bar{G}_i,M)$.
	Let $LB$ denote the optimal solution value for the remaining instance under the assumption that $a_{M,e} = U_{M,e}$ (and $a'_{M,e} = \gamma U_{M,e}$) for all $e \in \mmu$ and all not yet queried $M \in \mms \setminus (Q_i \cup \bar{G}_i)$.
	Clearly $LB \le \OPT$.
	
	Under this assumption, the remaining instance is again just a \setcover instance.
	Similar to the proof of~\Cref{lem:greedy:alpha}, we can argue that $\frac{A(Q_i \cup \bar{G}_i)}{LB} \le \ogs(Q_i\cup \bar{G}_i,\bar{M})$ as the optimal solution satisfies the remaining constraints at cost $LB$, but a single multiset can satisfy at most $\ogs(Q_i\cup \bar{G}_i,\bar{M})$ constraints.
	
	\jnew{	Observe that $\gs(Q_i\cup\bar{G}_i,\bar{M}) \ge \frac{1}{2} \ogs(Q_i\cup\bar{G}_i,\bar{M})$ holds by definition of the algorithm as $\bar{M}$ is the last added set in the iteration with $g = \ogs$ and iteration $i$ does not solve the instance.}
	By definition of $\gs$, this implies $A(Q_i \cup \bar{G}_i) - A(Q_i \cup \bar{G}_i \cup \bar{M}) = A(Q_i \cup \bar{G}_i) - A(Q_i \cup G_i) \ge \frac{1}{2} \cdot \ogs(Q_i \cup \bar{G}_i,\bar{M})$.
	In combination with $\frac{A(Q_i \cup \bar{G}_i)}{LB} \le \ogs(Q_i\cup \bar{G}_i,\bar{M})$, we get $A(Q_i \cup \bar{G}_i) - A(Q_i \cup G_i) \ge \frac{A(Q_i \cup \bar{G}_i)}{2 \cdot LB}$.
	This implies $A(Q_i \cup G_i) \le (1 - \frac{1}{2 \cdot LB}) \cdot A(Q_i \cup \bar{G}_i)$.
	As $\OPT \ge LB$ and $A(Q_i) \ge A(Q_i \cup \bar{G}_i)$, we can conclude $A(Q_i \cup G_i) \le (1 - \frac{1}{2 \cdot \OPT}) \cdot A(Q_i)$, which implies that $G_i$ $2$-approximates the greedy choice and, thus, satisfies the first property of \textsc{($\alpha,\beta,\gamma$)-Greedy}.
\end{proof}

\section{\minset under uncertainty}
\label{sec:minset}
We consider the general \minset under uncertainty.
In contrast to the previous section, we now also have uncertainty in the right-hand sides of~\eqref{eq:SetSelection}.
Since we consider the stochastic problem variant, recall that the balancing parameter is $\tau = \min_{I_i \in \ui} \tau_i$ with $\tau_i = \PR[w_i \ge \frac{U_i+L_i}{2}]$. 
Our goal is to iteratively add intervals from $\ui$ to the solution until it becomes feasible for~\eqref{eq:SetSelection}.
To that end, we prove the following main result.

\begin{theorem}
	\label{theo:minset}
	For $\tau>0$.
	There is an algorithm for \minset under uncertainty with a competitive ratio of $\mathcal{O}(\frac{1}{\tau} \cdot \log m \cdot \grsetu(\gamma)) \subseteq \mathcal{O}(\frac{1}{\tau} \cdot \log^2 m)$ with $\gamma = 2/s_{\min}$ for $s_{\min} = \min_{I_i \in \ui \colon (U_i-L_i) > 0} (U_i-L_i)$.
\end{theorem}

Exploiting~\Cref{theo:framework:minset}, we prove the statement by \nnew{providing} Algorithm~\ref{alg:selection1} \nnew{and showing that it} is an \alphagreedy algorithm for $\alpha = 2$, $\gamma = 2/s_{\min}$ and $\beta = \frac{1}{\tau} (\lceil\log_{1.5}(m \cdot (2/s_{\min}) \cdot \max_{I_i \in \ui} (U_i -L_i))\rceil  + \lceil\log_2(m)\rceil )$.
Note that $\alpha$ and $\gamma$ are defined as in the previous section for \minset with deterministic right-hand sides and will be used analogously.
For $\beta$ on the other hand, we require a larger value to adjust for the additional uncertainty in the right-hand sides $b_S = w^* - L_S$ for the uncertain $w^*$.
Notice that we do not have sufficient information to just execute Algorithm~\ref{alg:covering1} for \minset with deterministic right-hand sides as we need the right-hand side values to compute even the optimistic greedy values.

To handle this additional uncertainty, we want to ensure that each iteration of our algorithm $\alpha$-approximates the greedy choice for \emph{each} possible value of $w^*$.
To do so, we compute and query the best optimistic greedy choice for several carefully selected possible values of $w^*$.

To state our algorithm, we define a parametrized variant of~\eqref{eq:SetSelection} that states the problem under the assumptions that $w^* = w$ for some $w$ and that the set $Q \subseteq \ui$ has already been queried.
The coefficients are scaled
to $a'_i = (2/s_{\min})(w_i - L_i)$ and $b'_S(Q,w) =\max\{(2/s_{\min})(w - L_{S}) - \sum_{I_i \in Q \cap S} a'_i,0\}$.
As before, let $b'(Q,w)= \sum_{S \in \fs} b'_S(Q,w)$ denote the sum of right-hand sides.

\begin{equation}\tag{{\sc MinSetIP-Qw}}\label{eq:setsel}
	\begin{array}{llll}
		\min &\sum_{I_i \in \ui\setminus Q} x_i\\
		\text{s.t. }& \sum_{I_i \in S\setminus Q} x_i \cdot a'_i &\ge b'_S(Q,w) &\forall S \in \fs\\
		& x_i & \in \{0,1\}& \forall I_i \in \ui\\
	\end{array}
\end{equation}

As the right-hand sides are unknown, we define the greedy values for every possible value $w$ for $w^*$.
To that end, let $\gc(Q,I_i,w) = b'(Q,w)-b'(Q\cup \{I_i\},w)$ and $\gs(Q,I_i,w) = A(Q,w)-A(Q\cup \{I_i\},w)$, where $A(Q,w) = |\{S \in \fs \mid b'_S(Q,w) > 0\}|$ denotes the number of constrains in~\eqref{eq:setsel} that are not yet satisfied.
As before, $\gc(Q,I_i,w)$ and $\gs(Q,I_i,w)$ describe how much adding $I_i$ to the solution reduces the sum of right-hand sides and the number of non-satisfied constraints, respectively; now under the assumption that $w^* = w$.
For subsets $G \subseteq \ui\setminus Q$, we define the greedy values in the same way, i.e., $\gs(Q,G,w) = A(Q,w)-A(Q\cup G,w)$ and $\gc(Q,G,w) = b'(Q,w)-b'(Q\cup G,w)$.

Since our algorithm again does not have sufficient information to compute the precise greedy values $\gs(Q,I_i,w)$ and $\gc(Q,I_i,w)$ even for a fixed $w$, we again use the optimistic greedy values defined in the same way as in the previous section. That is
$$\ogc(Q,I_i,w) = {\sum_{S \in \fs \colon I_i \in S} b'_S(Q,w) - \max\{0, b'_S(Q,w)-\gamma(U_i-L_i)\}}$$ and $$\ogs(Q,I_i,w) = {|\{S\in \fs \colon I_i \in S \mid  b'_S(Q,w) > 0 \land b'_S(Q,w) - \gamma(U_i-L_i) \le 0\}|}.$$
For subsets $G \subseteq \ui\setminus Q$, the optimistic greedy values are defined analogously.


Similar to Algorithm~\ref{alg:covering1}, we would like to repeatedly compute and query the best optimistic greedy choice until the queried $I_i$ satisfies $w_i-L_i \ge \frac{U_i-L_i}{2}$ (cf. the repeat-statement).
However, we cannot decide which greedy value, $\ogc$ or $\ogs$, to use as deciding whether $b_S'(Q,w^*) < 1$ depends on the unknown $w^*$.
Instead, we compute and query the best optimistic greedy choice for both greedy values (cf.\ the for-loop). 
Even then, the best greedy choice still depends on the unknown right-hand sides.
Thus, we compute and query the best optimistic greedy choice for several carefully selected values $w$ (cf.\ the inner while-loop) to make sure that the queries of the iteration approximate the greedy choice for every possible $w^*$. Additionally, we want to ensure that we use at most $\beta$ queries in expectation within an iteration of the outer while-loop.

\begin{algorithm}[tb]
	\KwIn{Instance of \minset under uncertainty.}
	Scale all coefficients with $\gamma = 2/s_{\min}$ for $s_{\min} = \min_{I_i \in \ui \colon (U_i-L_i) > 0} (U_i-L_i)$\;
	$Q \gets \emptyset$, $w_{\min}$ $\gets$ minimum possible value $w^*$ (keep up-to-date)\;
	\While{the problem is not solved}{
		\ForEach{$g$ from the ordered list $\ogc,\ogs$}{
			$d \gets 1$; $Q' \gets Q$\;\label{line:selection1:startfirst}
			\lIf{$g=\ogc$}{$w_{\max} \gets$ max possible value $w^*$\label{line:setgreedymax1}}
			\lElse{$w_{\max} \gets$ max $w$ s.t.~$b_S'(Q,w) < 1$ for all $S \in \fs$}
			\While{$\exists w_{\min} \le w \le w_{\max}$ such that $\max_{I_h \in \ui\setminus Q} g(Q,I_h,w) \ge d$\label{line:selection1:while}}{
				\Repeat{$w_i - L_i \ge \frac{U_i - L_i}{2}$ or $\nexists w \le w_{\max} \colon \max_{I_h \in \ui\setminus Q} g(Q,w,I_h) \ge d$ \label{line:selection1:repeat}}{
					$w \gets$ min $w_{\min} \le w \le w_{\max}$ s.t.~$\max_{I_h \in \ui\setminus Q} g(Q,w,I_h) \ge d$ 
					\;\label{line:selection1:fsi1}
					$I_i \gets \arg\max_{I_h \in \ui\setminus Q} g(Q,I_h,w)$;
					Query $I_i$; $Q \gets Q \cup \{I_i\}$\;\label{line:selection1:query1}
					$Q_{1/2} \gets \{I_j \in Q\setminus Q' \mid w_j -L_j \ge \frac{U_j - L_j}{2}\}$\label{line:selection1:notation}\;
					\leIf{$g = \ogc$}{	$d \gets \gc(Q',Q_{1/2}, w)$}{$d \gets \gs(Q',Q_{1/2}, w)$\label{line:selection1:d}}	
				}	
			}
		}				
	}
	\caption{Algorithm for \minset under uncertainty.}
	\label{alg:selection1}
\end{algorithm}

To illustrate the ideas of the algorithm, consider an iteration of the outer while-loop. In particular, consider the for-loop iteration with $g=\ogs$ within this iteration.
Let $Q'$ denote the set of queries that were executed before the start of the iteration. 
Since we only care about the greedy value $\gc$ if there exists some $S \in \fs$ with $b_S'(Q') \ge 1$ (otherwise we use $\ogs$ and $\gs$ instead), we assume that this is the case. 
If not, we use a separate analysis for the for-loop iteration with $g = \ogs$.

Our goal for the iteration is to query a set of intervals $\bar{Q}$ that $2$-approximates the best greedy choice $I^*$ after querying $Q'$, i.e., it has a greedy value $\gc(Q',\bar{Q},w^*)\ge \frac{1}{2}\gc(Q',I^*,w^*)$ and, thus, satisfies~\Cref{lem:greedy:alpha}.
To achieve this for the unknown $w^*$, the algorithm uses the parameter $d$, which is initialized with $1$ (cf.~Line~\ref{line:selection1:startfirst}), the minimum possible value for $\ogc(Q',I^*,w^*)$ under the assumption that there exists some $S \in \fs$ with $b_S'(Q') \ge 1$.  
In an iteration of the inner while-loop, the algorithm repeatedly picks the minimal value $w$ such that the best current optimistic greedy choice has an optimistic greedy value of at least $d$ (cf.~Line~\ref{line:selection1:fsi1}). 
If no such value exists, then the loop terminates (cf.~Lines~\ref{line:selection1:while},~\ref{line:selection1:repeat}).
Afterwards, it queries the corresponding best optimistic greedy choice $I_i$ for the selected value $w$ (cf.~Line~\ref{line:selection1:query1}).
Similar to the algorithms of the previous section, this is done repeatedly until $w_i-L_i \ge (U_i - L_i)/2$.

The key idea to achieve the $2$-approximation with an expected number of queries that does not exceed $\beta$, is to always reset the value $d$ to $\gc(Q',Q_{1/2}, w)$, where $Q_{1/2}$ is the subset of all intervals $I_j$ that have already been queried in the current iteration of the outer while-loop and satisfy $w_j-L_j \ge (U_j - L_j)/2$ (cf.~Lines~\ref{line:selection1:notation},~\ref{line:selection1:d}).
This can be seen as an implicit doubling strategy to search for an unknown value. It leads to an exponential increase of $d$ over the iterations of the inner while-loop, which will allow us to bound their number.

With the following lemma, we prove that this choice of $d$ also ensures that the queries of the iteration indeed $2$-approximate the best greedy choice for $w^*$ if there exists a
$S \in \fs$ with $b_S'(Q',w^*) \ge 1$.
If there is no such set, we can use a similar proof w.r.t.~greedy value $\gs$.
For an iteration $j$ of the outer while-loop, let $G_j$ be the set of queries during the iteration and let $Q_j = \bigcup_{j' < j} G_j$ denote the queries before the iteration (cf.~$Q'$ in the algorithm).

\begin{lemma}
	\label{lem:main:claim2b}
	If there is an $S \in \fs$ with $b_S'(Q_j,w^*) \ge 1$, then $G_j$ $2$-approximates the greedy choice for the scaled instance with $w = w^*$ after querying $Q_j$.
\end{lemma}

\begin{proof}
	For an arbitrary but fixed realization, 
	consider an iteration $j$ of the outer while-loop such that  there exists a set $S \in \fs$ with $b_S'(Q_j,w^*) \ge 1$.
	
	Consider the subset $\bar{G}_j \subseteq G_j$ of queries that were executed with $g = \ogc$ before the increasing value $w$ (cf.~Line~\ref{line:selection1:fsi1}) surpasses $w^*$.
	That is, $\bar{G}_j$ only contains intervals that were queried for a current value $w \le w^*$.
	Let $\bar{I}_i$ be the element of $\bar{G}_j$ that is queried last.  
	Finally, let $\bar{d}_j$ denote the value $d$ computed by the algorithm in Line~\ref{line:selection1:d} directly after querying $\bar{I}_i$.	
	We continue to show that $G_j$ $2$-approximates the greedy choice of~\eqref{eq:setsel} for $Q= Q_j$ and $w = w^*$.

	Observe that $\gc(Q_j,\bar{G}_j,w^*) \ge  \bar{d}_j$.
	To see this, recall that $\bar{d}_j$ was computed in Line~\ref{line:selection1:d} after $\bar{I}_i$ was queried.
	Thus, $\bar{d}_j = \gc(Q',Q_{1/2}, w)$ for $Q' = Q_j$, $Q_{1/2} = \{I_j \in \bar{G}_j \mid w_j - L_j \ge \frac{U_j - L_j}{2} \}$ and some value $w$ with $w \le w^*$ by assumption.
	Since $w^* \ge w$ and $Q_{1/2} \subseteq \bar{G}_j$, the greedy value $\gc(Q_j,\bar{G}_j,w^*)$ can never be smaller than $\bar{d}_j = \gc(Q',Q_{1/2}, w)$. This implies $\gc(Q_j,\bar{G}_j,w^*) \ge \bar{d}_j$.
	
	We continue by showing that $d^* \le 2 \cdot \gc(Q_j,\bar{G}_j,w^*)$ holds for the best greedy value $d^*$ at the start of the iteration, i.e., $d^* = \max_{I_i \in \ui \setminus Q_j} \gc(Q_j,I_i,w^*)$.
	As $\bar{G}_j \subseteq G_j$, this implies $d^* \le 2 \cdot \gc(Q_j,G_j,w^*)$ and, thus, that $G_j$ satisfies~\Cref{def:greedyalpha}.

	To upper bound $d^*$, first observe that the best optimistic greedy value $d'$ after querying $\bar{G}_j \cup Q_j$ is smaller than $\bar{d}_j$, i.e., $d' = \max_{I_i \in \ui\setminus (Q_j \cup \bar{G}_j)} \ogc(Q_j \cup \bar{G}_j,I_i,w^*) < \bar{d}_j$.
	This follows directly from Line~\ref{line:selection1:fsi1} as $\bar{I}_i$ is the last query for a value $w \le w^*$ by assumption.
	As $\gc(Q_j,\bar{G}_j,w^*) \ge \bar{d}_j$, we get $\gc(Q_j,\bar{G}_j,w^*) \ge d'$.
	
	By definition of $\gc$, the best greedy value after querying $Q_j$ can never be larger than the sum of the greedy value of $\bar{G}_j$ after querying $Q_j$ and the best optimistic greedy value after querying $\bar{G}_j \cup Q_j$.
	Thus, we have $d^* \le \gc(Q_j,\bar{G}_j,w^*) + d' \le 2 \cdot \gc(Q_j,\bar{G}_j,w^*)$. This proves that $\bar{G}_j$ satisfies~\Cref{lem:greedy:alpha} and, thus, concludes this proof.
\end{proof}

Using a similar proof, 
\nnew{we} show the following lemma for the case where $b_S'(Q',w^*) < 1$ for all $S \in \fs$, which together with~\Cref{lem:main:claim2b} implies Property~$1$ of~\Cref{def:alphagreedy}.
While the main arguments remain the same as for the previous lemma, the more discrete nature of the greedy values $\gs$ and $\ogs$ poses several additional technical challenges that need to be taken care of.
For an iteration $j$ of the outer while-loop, let $G_j$ again be the set of queries during the iteration and let $Q_j = \bigcup_{j' < j} G_j$ denote the queries before the iteration (cf.~set~$Q'$ in the algorithm).

\begin{restatable}{lemma}{lemMinSetApprox}
	\label{lem:main:claim2a}
	If $b_S'(Q_j,w^*) < 1$ for all $S \in \fs$, then $G_j$ $2$-approximates the greedy choice for the scaled instance with $w = w^*$ after querying $Q_j$.
\end{restatable}

\begin{proof}		
	For an arbitrary but fixed realization, 
	consider an iteration $j$ of the outer while-loop such that $b_S'(Q_j,w^*) < 1$ for all $S \in \fs$.
	Our goal is to prove that $G_j$ approximates the greedy choice for the scaled instance with $w=w^*$ after querying $Q_j$ within a factor of two. That is, we have to prove $A(Q_j \cup G_j,w^*) \le (1-\frac{1}{2\cdot \opt}) \cdot A(Q_j,w^*)$. Recall that $A(Q_j,w^*)$ denotes the number of constraints that are not yet satisfied in the~\eqref{eq:setsel} instance for $Q = Q_j$ and $w = w^*$.
	
	Consider the subset $\bar{G}_j \subseteq G_j$ of queries to intervals $I_i$ that were executed with $g = \ogs$ for a current value $w \le w^*$ during iteration $j$ of the outer while-loop.
	Let $P_j \subseteq G_j$ denote the queries of the iteration that were executed before $\bar{G}_j$, i.e., that were executed during the iteration of the for-loop with $g = \ogc$.
	Note that $Q' = Q_j \cup P_j$ is the set of intervals queried before the beginning of the for-loop iteration with $g = \ogs$ during iteration $j$ of the outer while-loop.
	
	\textbf{Proof outline.} We start the proof by making some preliminary observations regarding greedy value $\gs$ and the scaling factor $\gamma$ that will be crucial for the remainder of the proof. Then we proceed by proving that $\bar{G}_j$ approximates the greedy choice of~\eqref{eq:setsel} for $Q= Q_j \cup P_j$ and $w = w^*$ within a factor of $2$. To that end, we first derive a lower bound on the greedy value $\gs(Q_j\cup P_j, \bar{G}_j, w^*)$ and afterwards compare this lower bound with $\opt$.
	Finally, we use the fact that $\bar{G}_j$ $2$-approximates the greedy choice after querying $Q_j \cup P_j$ to show that $G_j$ approximates the greedy choice of~\eqref{eq:setsel} for $Q= Q_j$ and $w = w^*$ within a factor of $2$.
	
	\textbf{Preliminary observations.} Before we start with the proof, recall that an interval $I_i$ with $(w_i-L_i) \ge \frac{1}{2} \cdot (U_i-L_i)$ satisfies $a'_{i} = \frac{2 \cdot (w_i - L_i)}{s_{\min}} \ge \frac{U_i-L_i}{s_{\min}} \ge 1$ by choice of the scaling parameter $\gamma = \frac{2}{s_{\min}}$.
	This implies that adding $I_i$ to the solution satisfies all constraints for sets $S$ with $I_i \in S$ 
	as long as we are considering values $w$ with $b'_S(Q_j,w) \le 1$ for all $S \in \fs$.
	Thus, for such intervals and values $w$, the greedy value $\gs$ of $I_i$ is then equal to the optimistic greedy value $\ogs$ as even under the assumption $a_i' = \gamma (U_i-L_i)$ adding interval $I_i$ cannot satisfy more constraints. 
	By assumption, this in particular holds for all values $w \le w^*$.	
	This also means that the greedy values $\gs$ and $\ogs$ of such intervals $I_i$ only increase with an increasing value $w$, as long as $b'_S(Q_j,w) \le 1$ still holds for all $S \in \fs$.

	\textbf{Lower bound on $\gs(Q_j\cup P_j, \bar{G}_j, w^*)$.} We continue by deriving a lower bound on $\gs(Q_j\cup P_j, \bar{G}_j, w^*)$. 
	Let $\bar{I}_i$ be the element of $\bar{G}_j$ that is queried last and  
	let $\bar{d}_j$ denote the value $d$ computed by the algorithm in Line~\ref{line:selection1:d} directly after querying $\bar{I}_i$. We first observe that $\gs(Q_j\cup P_j, \bar{G}_j, w^*) \ge  \bar{d}_j$.
	To see this, recall that $\bar{d}_j$ was computed in Line~\ref{line:selection1:d} after $\bar{I}_i$ was queried.
	Thus, $\bar{d}_j = \gs(Q',Q_{1/2}, w)$ for $Q' = Q_j \cup P_j$, $Q_{1/2} = \{I_j \in \bar{G}_j \mid w_j - L_j \ge \frac{U_j - L_j}{2} \}$ and some value $w$ with $w \le w^*$ by assumption.
	Since $w^* \ge w$ and $Q_{1/2}\subseteq \bar{G}_j$, the greedy value $\gs(Q_j \cup P_j,\bar{G}_j,w^*)$ can never be smaller than $\bar{d}_j = \gs(Q',Q_{1/2}, w)$ according to the observations stated at the beginning of the proof.
	This implies $\gs(Q_j\cup P_j,\bar{G}_j,w^*) \ge \bar{d}_j$.
	
	Assume for now that the algorithm queried $Q_s = \bar{G}_j \setminus Q_{1/2}$ before $Q_{1/2}$. 
	We consider the best optimistic greedy value $d^*$ after $Q_j \cup P_j \cup Q_s$ has already been queried, i.e, $d^* = \max_{I_i \in \ui \setminus (Q_j\cup P_j \cup Q_s)} \ogs(Q_j\cup P_j,I_i,w^*)$.
	To bound $d^*$, first observe that the best optimistic greedy value $d'$ after querying $\bar{G}_j \cup Q_j \cup P_j$ is smaller than $\bar{d}_j$, i.e., $d' = \max_{I_i \in \ui\setminus (Q_j \cup P_j \cup \bar{G}_j)} \ogs(Q_j \cup P_j \cup \bar{G}_j,I_i,w^*) < \bar{d}_j$.
	This follows directly from Line~\ref{line:selection1:fsi1} as $\bar{I}_i$ is the last query for a value $w \le w^*$ by assumption.
	Since we already showed $\gs(Q_j\cup P_j,\bar{G}_j,w^*) \ge \bar{d}_j$, we get $\gs(Q_j \cup P_j,\bar{G}_j,w^*) \ge d'$.
	
	By definition of $\ogs$, definition of $Q_{1/2}$, and the assumption that we only consider values $w$ with $b'_S(Q',w) < 1$ for all $S \in \fs$, the best optimistic greedy value after querying $Q_j \cup P_j \cup Q_s$ 
	can never be larger than the sum of the optimistic greedy value of $Q_{1/2}$ after querying $Q_j \cup P_j \cup Q_s$ and the best optimistic greedy value after querying $Q_j \cup P_j \cup Q_s \cup Q_{1/2} = Q_j \cup P_j \cup \bar{G}_j$.
	By assumption that we only consider values $w$ with $b'_S(Q',w) < 1$ for all $S \in \fs$, we have $\ogs(Q_j\cup P_j\cup Q_s,Q_{1/2},w^*) = \gs(Q_j\cup P_j\cup Q_s,Q_{1/2},w^*)$ (as argued at the beginning of the proof).	
	Putting it together, we get 
	\begin{align*}
		d^* &\le \ogs(Q_j\cup P_j\cup Q_s,Q_{1/2},w^*) + d' \\  &= \gs(Q_j\cup P_j\cup Q_s,Q_{1/2},w^*) + d' \\  &\le \gs(Q_j\cup P_j\cup Q_s,Q_{1/2},w^*) +  \gs(Q_j\cup P_j,\bar{G}_j,w^*) \\ 
	\end{align*}
	
	\textbf{Proving that $\bar{G}_j$ approximates its greedy choice.} We continue the proof by using the inequality for $d^*$ in order to show that $\bar{G}_j$ approximates its greedy choice.
	
	To that end, we consider a relaxed instance $R$ of~\eqref{eq:setsel} with $Q = Q_j \cup P_j$ and $w = w^*$, i.e., we assume that $Q_j \cup P_j$ has already been queried and consider the right-hand sides $b'_S(Q_j \cup P_j,w^*)$ for all $S \in \fs$.
	We relax the instance by increasing the left-hand side coefficients and decreasing the cost coefficients. Let $\hat{a}$ denote the relaxed coefficients.
	First consider the intervals in $\bar{G_j}$. For these intervals, we use the original scaled coefficients. That is, we set $\hat{a}_i = a_i' = \gamma(w_i-L_i)$. Recall that $Q_{1/2} = \{I_i \in \bar{G}_j \mid w_i-L_i \ge \frac{1}{2} (U_i-L_i)\}$ and $Q_s = \bar{G}_j \setminus Q_{1/2}$. For intervals $I_i \in Q_s$, we set the cost coefficients to zero. Thus, these intervals can be added to any solution without increasing the objective value. All other intervals keep their cost coefficients of one.
	For all other intervals $I_i \in \ui \setminus (Q_j\cup P_j \cup \bar{G}_j)$, we set $\hat{a}_i = \gamma(U_i - L_i)$. That is, we assume that the coefficients of these intervals are slightly larger than their largest possible value. 
	Since instance $R$ compared to the original instance only increases left-hand side coefficients and decreases cost coefficients, it clearly is a relaxation and, thus, $LB \le \opt$ for the optimal objective value $LB$ of the relaxed instance $R$.
	
	Since we assume $b'_S(Q_j \cup P_j,w^*) < 1$ for all $S \in \fs$, the right-hand sides of instance $R$ are all strictly smaller than one. Furthermore, by choice of the scaling parameter $\gamma$, all intervals $I_i$ satisfy $\gamma (U_i-L_i) \ge 1$. Thus, the intervals $I_i$ in $Q\setminus (Q_j \cup P_j \cup \bar{G}_j)$ have coefficients $\hat{a}_i \ge 1$. By the observations at the beginning of the proof, the same holds for the intervals in $Q_{1/2}$. Only the intervals in $Q_s$ can potentially have coefficients smaller than the right-hand sides. This also means that the greedy value $\gs$ of an interval $I_i \in Q\setminus (Q_j \cup P_j \cup Q_s)$ for instance $R$ can never be increased by previously adding intervals of $Q_s$ to the solution, since adding $I_i$ already satisfies all constraints for sets $S$ with $I_i \in S$. For intervals in $Q_s$, it might be the case that previously adding further elements of $Q_s$ to the solution increases the greedy value. This is, because the coefficient of an $I_i \in Q_s$ might be too small to satisfy a constraint $S$ with $I_i \in S$ but previously adding further intervals from $Q_s$ might decrease the right-hand side of $S$ enough such that adding $I_i$ can satisfy the constraint and, thus, increase the greedy value of~$I_i$.
	
	Now, consider the number of constraints that are initially not satisfied in $R$. This number is exactly the same as in the non-relaxed original instance after querying $Q_j \cup P_j$ since $R$ uses the exact same right-hand sides. Thus, the number of not yet satisfied constraints in $R$ is $A(Q_j \cup P_j,w^*)$. Let $\hat{d} = \max_{I_i \in \ui\setminus(Q_j \cup P_j \cup Q_s)} |\{S \in S \mid b'_S(Q_j\cup P_j,w^*) > 0 \land b'_S(Q_j\cup P_j,w^*) - \hat{a}_i - \sum_{I_h \in Q_s} \hat{a}_h  \le 0 \}|$ denote the maximum number of constraints that can be satisfied by adding $Q_s$ and \emph{one} interval $I_i$ in $Q_j \cup P_j \cup Q_s$ to the solution for $R$. In a sense, $\hat{d}$ is the best possible greedy value $\gs$ for instance $R$ that can be achieve by adding such a set $Q_s \cup \{I_i\}$ to the solution. 
	Remember that the elements of $Q_s$ do not incur any costs, so $\hat{d}$ constraints can be satisfied at cost one. 
	Following the arguments in the proof of~\Cref{lem:greedy:alpha}, this implies $\hat{d} \ge \frac{A(Q_j \cup P_j,w^*)}{LB}$ as the optimal solution satisfies $A(Q_j \cup P_j,w^*)$ constraints at cost $LB$ but it is impossible to satisfy more than $\hat{d}$ constraints with cost one. This last implication crucially uses the observation that the greedy values of intervals in $Q\setminus (Q_j \cup P_j \cup Q_s)$ cannot increase by previously adding other intervals to the solution. Because this is the case, even adding an interval later cannot satisfy more than $\hat{d}$ constraints, which gives us $\hat{d} \ge \frac{A(Q_j \cup P_j,w^*)}{LB}$. We remark that this also is the reason why the offline Algorithm~\ref{alg:covering1} starts with greedy value $\gc$ and only switches to greedy value $\gs$ once the instance reduced to a \setcover instance.
	
	Taking a closer look at $\hat{d}$, we can observe that this greedy value is exactly the sum of the greedy values $\gs(Q_j \cup P_j, Q_s,w^*)$ and $d^* = \max_{I_i \in Q\setminus(Q_j \cup P_j\cup Q_s)} \ogs(Q_j\cup P_j \cup Q_s, I_i, w^*)$ for the non-relaxed instance. This is, because the original instance uses the same right-hand sides as $R$, the coefficients of intervals in $Q_s$ are the same in both instances, and the optimistic greedy value $\ogs(Q_j\cup P_j \cup Q_s, I_i, w^*)$ like instance $R$ already assumes that all intervals in $Q\setminus (Q_j\cup P_j \cup Q_s)$ have coefficients larger than all right-hand sides.
	Thus, these three arguments imply $\hat{d} = \gs(Q_j \cup P_j, Q_s,w^*) + d^*$.
	Combining this equality with $\hat{d} \ge \frac{A(Q_j \cup P_j,w^*)}{LB}$ and the previously derived upper bound on $d^*$ gives us:
	\begin{align*}
		\frac{A(Q_j \cup P_j,w^*)}{LB} &\le \hat{d}\\
		&=  \gs(Q_j \cup P_j, Q_s,w^*) + d^*\\
		&\le\gs(Q_j \cup P_j, Q_s,w^*) + \gs(Q_j\cup P_j\cup Q_s,Q_{1/2},w^*) + \gs(Q_j \cup P_j, \bar{G}_j,w^*)\\
		&= 2 \cdot \gs(Q_j \cup P_j, \bar{G}_j,w^*).
	\end{align*}
	Here the second equality uses that
	\begin{align*}
		&\gs(Q_j \cup P_j, Q_s,w^*) + \gs(Q_j\cup P_j\cup Q_s,Q_{1/2},w^*)\\
		=&(A(Q_j \cup P_j, w^*) - A(Q_j \cup P_j \cup Q_s,w^*))\\
		+& (A(Q_j\cup P_j\cup Q_s, w^*)- A(Q_j\cup P_j\cup Q_s \cup Q_{1/2}, w^*))\\
		=& A(Q_j \cup P_j, w^*) - A(Q_j\cup P_j\cup Q_s \cup Q_{1/2}, w^*)\\
		=& A(Q_j \cup P_j, w^*) - A(Q_j\cup P_j\cup \bar{G}_j, w^*)\\
		=& \gs(Q_j\cup P_j,\bar{G}_j,w^*).\\
	\end{align*}
	
	Plugging in the definition of $\gs$ and the inequality $LB \le \OPT$ yields
	\begin{align*}
		&\frac{A(Q_j \cup P_j,w^*)}{LB} \le 2 \cdot (A(Q_j \cup P_j,w^*) - A(Q_j \cup P_j \cup \bar{G}_j,w^*))\\
		\Leftrightarrow& A(Q_j \cup P_j \cup \bar{G}_j,w^*) \le A(Q_j \cup P_j,w^*) \cdot (1-\frac{1}{2\cdot LB})\\
		\Rightarrow&  A(Q_j \cup P_j \cup \bar{G}_j,w^*) \le A(Q_j \cup P_j,w^*) \cdot (1-\frac{1}{2 \cdot \OPT}).\\
	\end{align*}
	Thus, $\bar{G}_j$ approximates the greedy choice of~\eqref{eq:setsel} for $Q = Q_j\cup P_j$ and $w=w^*$ within a factor of $2$. 
	
	\textbf{Concluding the proof.} Remember that our goal was to show that $G_j$ approximates the greedy choice for $w=w^*$ after querying $Q_j$ within a factor of $2$. 
	To that end, observe that $P_j \cup \bar{G}_j \subseteq G_j$ implies $A(Q_j \cup P_j \cup \bar{G}_j,w^*) \ge A(Q_j \cup G_j,w^*)$ and that $Q_j \subseteq Q_j \cup P_j$ implies $A(Q_j,w^*)  \ge A(Q_j\cup P_j,w^*)$. 
	Plugging these inequalities into the previously derived inequality for $A(Q_j \cup P_j \cup \bar{G}_j,w^*)$ yields $A(Q_j \cup G_j,w^*) \le A(Q_j,w^*) \cdot (1-\frac{1}{2 \cdot \OPT})$.
	This means that $G_j$ also $2$-approximates the greedy choice after querying $Q_j$ for $w=w^*$, which concludes the proof.
\end{proof}

Since the~\Cref{lem:main:claim2a,lem:main:claim2b} imply Condition~$1$~of~\Cref{def:alphagreedy}, it remains to show Condition~$2$ in order to apply~\Cref{theo:framework:minset}. 
The proof idea is to show that parameter $d$ increases by a factor of at least $1.5$ in each iteration of the inner while-loop with $g = \ogc$.
As $\ogc(Q,I_i,w)$ is upper bounded by $m (2/s_{\min})\max_{I_i \in \ui} (U_i -L_i)$, this means the inner loop executes at most$\lceil\log_{1.5}(m (2/s_{\min})\max_{I_i \in \ui} (U_i -L_i))\rceil$ iterations for $g = \ogc$.
For $g=\ogs$, we can argue in a similar way that at most $\lceil\log_2(m)\rceil$ iterations are executed.
Similar to the previous section on \minset with deterministic right-hand sides, we can also show that each iteration of the inner while-loop executes at most $\frac{1}{\tau}$ queries in expectation.
Combining these insights, we can bound the expected number of queries during an execution of the outer while-loop by $\beta$.
Formally proving the increase of $d$ by $1.5$ requires to take care of several technical challenges. 
The basic idea is to exploit that the interval $I_i$ queried in Line~\ref{line:selection1:query1} has an optimistic greedy value of at least $d$.
If it satisfies $w_i-L_i \ge \frac{1}{2}(U_i-L_i)$, we show that the actual greedy value is at least $d/2$.
When $d$ is recomputed in Line~\ref{line:selection1:d}, then $I_i$ is a new member of the set $Q_{1/2}$ and leads to the increase of $d/2$. 

We conclude the section by formalizing this proof idea.
For each iteration $j$ of the outer while-loop, let $X_j$ be a random variable denoting the number of queries executed during iteration $j$ and let $Y_j$ be a variable indicating whether the iteration is executed ($Y_j = 1$) or not ($Y_j = 0$). We prove the following lemma, which implies that the algorithm also satisfies  the second condition of~\Cref{def:alphagreedy} and, thus, satisfies~\Cref{theo:minset}.

\begin{lemma}
	\label{lem:main:claim1}
	$\EX[X_j \mid Y_j = 1] \le \beta = \frac{1}{\tau} \cdot (\lceil\log_{1.5}(m \cdot \frac{2 \cdot \max_{I_i \in \ui} (U_i -L_i)}{s_{\min}})\rceil  + \lceil\log_{2}(m)\rceil)$.
\end{lemma}
\begin{proof}
	Consider an iteration $j$ of the outer while-loop of Algorithm~\ref{alg:selection1}.
	In the following, all probabilities and expected values are under the condition $Y_j = 1$.
	
	For iteration $j$, let $A_{jk}$ and $B_{jk}$ be random variables denoting the number of queries in iteration $k$ of the inner while-loop for $g = \ogc$ and $g= \ogs$, respectively.
	Then, $\EX[X_j] = \sum_{k} \EX[B_{jk}] + \sum_{k} \EX[A_{jk}]$.
	
	Let $\bar{A}_{jk}$ and $\bar{B}_{jk}$ be random variables indicating whether iteration $k$ of the inner while-loop with $g = \ogc$ and $g=\ogs$, respectively, is executed ($\bar{A}_{jk} = 1$, $\bar{B}_{jk} = 1$) or not.
	Given that such an iteration is executed, we can exploit that the termination criterion $w_i \ge \frac{1}{2} \cdot (U_i-L_i)$ occurs with probability $\PR[(w_i-L_i) \ge (U_i-L_i)/2] \ge \PR[w_i \ge (U_i+L_i)/2]   \ge \tau$ to show $\EX[A_{jk} \mid \bar{A}_{jk} = 1] \le \frac{1}{\tau}$
	and $\EX[B_{jk} \mid \bar{B}_{jk} = 1]\le \frac{1}{\tau}$.
	Using these bounds and exploiting that $\EX[A_{jk}\mid \bar{A}_{jk}=0]=\EX[B_{jk}\mid \bar{B}_{jk}=0]= 0$, we get
	\begin{align*}
		\EX[X_j]
		&= \sum_k (\EX[A_{jk}] + \EX[B_{jk}])\\
		&= \sum_k \PR[\bar{A}_{jk} = 1] \cdot \EX[A_{jk} \mid \bar{A}_{jk}=1] 
		+ \PR[\bar{B}_{jk} = 0 ] \cdot \EX[B_{jk} \mid \bar{A}_{jk}=0]\\
		&+ \sum_k \PR[\bar{B}_{jk} = 1] \cdot \EX[B_{jk} \mid \bar{B}_{jk}=1]
		+ \PR[\bar{B}_{jk} = 0] \cdot \EX[B_{jk} \mid \bar{B}_{jk}=0]\\
		&= \sum_k \PR[\bar{A}_{jk} = 1] \cdot \EX[A_{jk} \mid \bar{A}_{jk}=1] + 
		\sum_k \PR[\bar{B}_{jk} = 1] \cdot \EX[B_{jk} \mid \bar{B}_{jk}=1] \\
		&\le \frac{1}{\tau} \cdot \left(\sum_k \PR[\bar{A}_{jk} = 1] + \sum_k \PR[\bar{B}_{jk} = 1] \right).
	\end{align*}
	
	Thus, it only remains to bound $\sum_k \PR[\bar{A}_{jk} = 1] + \sum_k \PR[\bar{B}_{jk} = 1]$.
	We do this by separately proving $\sum_k \PR[\bar{A}_{jk} = 1] \le \lceil\log_{1.5}\left(m \cdot (2/s_{\min}) \cdot \max_{I_i \in \ui} (U_i -L_i) \right)\rceil$ and $\sum_k \PR[\bar{B}_{jk} = 1] \le \lceil\log_2(m)\rceil$.
	
	Putting all bounds together, we then get
	$
	\EX[X_j] \le \frac{1}{\tau} \cdot (\sum_k \PR[\bar{A}_{jk} = 1] + \PR[\bar{B}_{jk} = 1 ] )
	\le \frac{1}{\tau} \cdot (\lceil\log_{1.5}(m \cdot \frac{2 \cdot \max_{I_i \in \ui} (U_i -L_i)}{s_{\min}})\rceil+\lceil\log_2(m)\rceil).
	$

	\textbf{Upper bound for $\sum_k \PR[\bar{A}_{jk} = 1]$.}
	Note that $\sum_k \PR[\bar{A}_{jk} = 1]$ is just the expected number of iterations of the inner while-loop with $g = \ogc$ during iteration $j$ of the outer while-loop. 
	We bound this number by showing that the value $d$ increases by a factor of at least $1.5$ in each iteration except possibly the first and last ones.
	As $d$ is upper bounded by $a = (2/s_{\min}) \cdot m \cdot \max_{I_i \in \ui} (U_i - L_i)$, the number of iterations then is at most $\lceil\log_{1.5}\left(a\right)\rceil$ for every realization of values.

	To prove that $d$ increases by a factor of $1.5$ in each iteration of the inner while-loop, consider the last execution of the repeat-statement within an iteration of the inner while-loop with $g=\ogc$ (that is not the first or last one). 
	Let $I_i$ be the interval queried in that last iteration of the repeat-statement, let $Q$ denote the set of $I_i$ and all intervals that were queried before $I_i$, let $Q'$ denote the set of all intervals queried before the current execution of the inner while-loop was started, and let $Q_{1/2} = \{I_j \in Q \setminus Q' \mid w_j -L_j \ge \frac{U_j - L_j}{2}\}$.
	Furthermore, let $w$ denote the value of Line~\ref{line:selection1:fsi1} computed before $I_i$ was queried in the following line.
	Note that this notation matches the one used in the algorithm at the execution of Line~\ref{line:selection1:notation} directly after $I_i$ was queried.
	
	Let $\bar{d}$ denote the value $d$ computed by the algorithm \emph{before} $I_i$ was queried. 
	That is, $\bar{d} = \gc(Q',Q_{1/2} \setminus \{I_i\}, w')$ for the value $w'$ that was computed in the previous execution of Line~\ref{line:selection1:fsi1}.
	By choice of $I_i$, we have $\ogc(Q\setminus \{I_i\},I_i,w) \ge \bar{d}$ and $\gc(Q\setminus \{I_i\},I_i,w) \ge \bar{d}/2$ (since $w_i-L_i \ge (U_i-L_i)/2$).
	
	The value $d$ computed after querying $I_i$ is defined as 
	\begin{align*}
		d = \gc(Q',Q_{1/2}, w) &\ge \gc(Q',Q_{1/2}\setminus\{I_i\}, w) + \gc(Q' \cup Q_{1/2} \setminus \{I_i\}, I_i,w)\\ &\ge \gc(Q',Q_{1/2}\setminus\{I_i\}, w) + \gc(Q \setminus \{I_i\}, I_i,w),
	\end{align*}
	where the first inequality holds by definition of $\gc$ and the second inequality holds because $Q\setminus\{I_i\}\supseteq Q'\cup Q_{1/2} \setminus \{I_i\}$ implies $\gc(Q' \cup Q_{1/2} \setminus \{I_i\}, I_i,w) \ge \gc(Q \setminus \{I_i\}, I_i,w)$.
	
	In an iteration of the for-loop with $g = \gc$, the current value $w$ only increases, which implies $w' \le w$.
	Since the greedy value $\gc$ only increases for increasing values $w$, we get  $\gc(Q',Q_{1/2}\setminus\{I_i\}, w) \ge \gc(Q',Q_{1/2}\setminus \{I_i\}, w')$ and, therefore $\gc(Q',Q_{1/2}, w) \ge \gc(Q',Q_{1/2} \setminus \{I_i\}, w') + \gc(Q\setminus \{I_i\}, I_i,w)$.
	Plugging in $\bar{d} = \gc(Q',Q_{1/2}\setminus \{I_i\}, w')$ and $\gc(Q\setminus \{I_i\} ,I_i,w) \ge \bar{d}/2$ yields $d \ge 1.5 \cdot \bar{d}$.
	
	Note that this only shows that $d$ increases by a factor of $1.5$ compared to $\bar{d}$, the old value computed in the previous execution of the repeat statement.
	Our goal was to show that $d$ increases by factor $1.5$ compared to the previous iteration of the inner-while loop.
	However, since the value $w$ only increases, the greedy values $\gc$ only increase.
	This implies that $\bar{d}$ can only be larger than the corresponding value at the end of the previous iteration of the inner while-loop.
	Thus, we have shown that $d$ increases by a factor of $1.5$ compared to its value at the end of the previous iteration of the inner-while loop. 
	This implies $\sum_k \PR[\bar{A}_{jk} = 1] \le \lceil\log_{1.5}\left(m \cdot (2 \cdot \max_{I_i \in \ui} (U_i -L_i))/s_{\min} \right)\rceil$.
	
	\textbf{Upper bound for $\sum_k \PR[\bar{B}_{jk} = 1]$.}
	Next, we use similar arguments to show $\sum_k \PR[\bar{B}_{jk} = 1] \le \lceil\log_2(m)\rceil$.
	Consider an iteration $j$ of the outer while-loop.
	
	We first show that, assuming $g = \ogs$, the value $d$ increases by a factor of at least $2$ in each iteration of the inner while-loop (except possibly the first and last ones) during iteration $j$ of the outer while-loop.
	
	By Line~\ref{line:setgreedymax1}, iterations with $g = \ogs$ only consider values $w$ with $b'_S(Q',w) < 1$ for all $S \in \fs$, where $Q'$ is the set of intervals queried before the start of the current execution of the inner while-loop.
	For an interval $I_i$ with $(w_i-L_i) \ge \frac{1}{2} \cdot (U_i-L_i)$, this means $a'_{i} = \frac{2 \cdot (w_i - L_i)}{s_{s_{\min}}} \ge \frac{U_i-L_i}{s_{\min}} \ge 1$.
	This implies that adding $I_i$ to the solution satisfies all constraints for sets $S$ with $I_i \in S$ (at least for all values $w$ with $b'_S(Q',w) < 1$ for all $S \in \fs$). 	
	Thus, the greedy value $\gs$ of $I_i$ is then equal to the optimistic greedy value $\ogs$ as even if $a_i = U_i-L_i$ the interval cannot satisfy more constraints.
	
	Furthermore, as long as we only consider values $w$ with $b'_S(Q',w) < 1$ for all $S \in \fs$, the greedy values $\gs$ and $\ogs$ of intervals $I_i$ with  $(w_i-L_i) \ge \frac{1}{2} \cdot (U_i-L_i)$ can only increase with increasing $w$ and never increase for decreasing $w$. The reason for this is that such an increase in $w$ can only lead to \emph{more} constraints becoming \emph{not} satisfied for the right-hand sides $b'_S(Q',w)$. Thus, the number of constraints that adding $I_i$ can satisfy only increases.
	This also means, that the value $w$ as used by the algorithm will never decrease during the current iteration of the outer while-loop with $g= \ogs$.
	
	Using these observations, we can essentially repeat the proof of the previous case to show that the value $d$ increases by a factor of at least $2$ in each iteration of the inner while-loop with $g=\ogs$ (except possibly the first and last ones).
	
	Let $I_i$ be the interval queried in the last iteration of the repeat-statement within such an iteration of the inner while-loop, let $Q$ denote the set of $I_i$ and all intervals that were queried before $I_i$, let $Q'$ denote the set of all intervals queried before the current execution of the inner while-loop was started, and let $Q_{1/2} = \{I_j \in Q \setminus Q' \mid w_j -L_j \ge \frac{U_j - L_j}{2}\}$.
	Furthermore, let $w$ denote the value of Line~\ref{line:selection1:fsi1} computed before $I_i$ was queried in the following line.
	Note that this notation matches the one used in the algorithm at the execution of Line~\ref{line:selection1:notation} directly after $I_i$ was queried.
	
	Let $\bar{d}$ denote the value $d$ computed by the algorithm \emph{before} $I_i$ was queried. 
	That is, $\bar{d} = \gs(Q',Q_{1/2} \setminus \{I_i\}, w')$ for the value $w'$ that was computed in the previous execution of Line~\ref{line:selection1:fsi1}.
	By choice of $I_i$, we have $\ogs(Q\setminus \{I_i\},I_i,w) \ge \bar{d}$ and $\gs(Q\setminus \{I_i\},I_i,w) =  \ogs(Q\setminus \{I_i\},I_i,w) \ge \bar{d}$ (since $w_i-L_i \ge (U_i-L_i)/2$ and as argued above).
	
	The value $d$ computed after querying $I_i$ is defined as 
	\begin{align*}
		d = \gs(Q',Q_{1/2}, w) &\ge \gs(Q',Q_{1/2}\setminus\{I_i\}, w) + \gs(Q' \cup Q_{1/2} \setminus \{I_i\}, I_i,w)\\
		&\ge \gs(Q',Q_{1/2}\setminus\{I_i\}, w) + \gs(Q\setminus \{I_i\}, I_i,w).
	\end{align*}
	Here, the first inequality holds by definition of $\gs$, definition of $Q_{1/2}$, and by the assumption that we only consider values $w$ with $b'_S(Q',w) < 1$ for all $S \in \fs$.
	The second inequality holds because $Q\setminus\{I_i\}\supseteq Q'\cup Q_{1/2} \setminus \{I_i\}$ implies $\gs(Q' \cup Q_{1/2} \setminus \{I_i\}, I_i,w) \ge \gs(Q \setminus \{I_i\}, I_i,w)$. Note that this implication only holds under the assumption that $b'_S(Q',w) < 1$ for all $S \in \fs$ for intervals $I_i$ with $a_i' \ge 1$.
	
	As argued above, the value $w$ used by the algorithm only increases, i.e., $w' \le w$.
	Furthermore, again as argued above, the greedy value $\gs$ of intervals $I_j$ with $a_j' \ge 1$ only increases for increasing values $w$.
	Thus, we get  $\gs(Q',Q_{1/2}\setminus\{I_i\}, w) \ge \gs(Q',Q_{1/2}\setminus \{I_i\}, w')$ and, therefore, $\gs(Q',Q_{1/2}, w) \ge \gs(Q',Q_{1/2} \setminus \{I_i\}, w') + \gs(Q\setminus \{I_i\}, I_i,w)$.
	Plugging in $\bar{d} = \gs(Q',Q_{1/2}\setminus \{I_i\}, w')$ and $\gs(Q\setminus \{I_i\} ,I_i,w) \ge \bar{d}$ yields $d \ge 2 \cdot \bar{d}$.
	This implies that $d$ increased by a factor of at least $2$ compared to its old value at the end of the previous iteration of the inner while-loop.
	As the greedy value $\gs$ is upper bounded by $m$, we get  $\sum_k \PR[\bar{B}_{jk} = 1] \le \lceil\log_2(m)\rceil$.
\end{proof}

The \Cref{lem:main:claim1,lem:main:claim2a,lem:main:claim2b} imply that Algorithm~\ref{alg:selection1} satisfies~\Cref{def:alphagreedy} for $\alpha = 2$, $\gamma = 2/s_{\min}$ and $\beta = \frac{1}{\tau} (\lceil\log_{1.5}(m \cdot 2  (\max_{I_i \in \ui} (U_i -L_i))/s_{\min})\rceil  + \lceil\log_2(m)\rceil )$.
Thus,~\Cref{theo:framework:minset} implies~\Cref{theo:minset}.

\section{Final remarks}

In this paper, we provide the first results for \minset under stochastic explorable uncertainty by exploiting a connection to a covering problem and extending techniques for solving covering problems to our setting with uncertainty. 

Since our results, in expectation, break adversarial lower bound instances for a number of interesting combinatorial problems under explorable uncertainty, e.g., matching, knapsack, solving ILPs~\cite{meissner18querythesis}, we hope that our techniques lay the foundation for solving more general problems. 
In particular if we consider the mentioned problems with uncertainty in the cost coefficients and our goal is to query elements until we can identify an optimal solution to the underlying problem, then all these problems admit the same connection to covering problems as \minset and can also be written as covering ILPs with uncertain coefficients and right-hand sides. The difference to \minset is that the number of constraints in the covering representation for these problems might be exponential in the input size of the underlying optimization problem. In a sense, each feasible solution for the underlying optimization problem would define a constraint in the covering representation. This means that using the results of this paper as a blackbox to solve such optimization problems under explorable uncertainty does not yield sublinear competitive ratios as the number of constraints becomes too large. 
For future research, we suggest to investigate whether one can exploit the additional structure in these constraints to still achieve improved competitive ratios. 

With respect to the results of this paper, we leave open whether the second $\log$ factor in our main result, \Cref{theo:minset}, is necessary. Furthermore, the best competitive ratio achievable in exponential running time also remains open. Finally, we remark that we expect our algorithms to be parameterizable by the choice of the balancing parameter. We defined the parameter as the probability that the precise value is larger than the center of the corresponding interval. Alternatively, one could pick any fraction of the interval, say one third, as a threshold and define the parameter as the probability that the precise value lies outside the first third of its interval. As the only parts of our algorithm that use the definition of the balancing parameter are the termination criteria of the repeat-statements, the definition of the scaling factor $\gamma$ and the definition of the set $Q_{1/2}$ in Algorithm~\ref{alg:selection1}, we expect that adjusting these parts to a different balancing parameter suffices to make our algorithms work for such parameters. The choice of the threshold for the balancing parameter would then influence the constant factor within the competitive ratios. 

%
%
%
%
%
%

%

%
%
\bibliographystyle{plain}
\bibliography{bib}

\begin{thebibliography}{10}

\bibitem{AgarwalAK19}
Arpit Agarwal, Sepehr Assadi, and Sanjeev Khanna.
\newblock Stochastic submodular cover with limited adaptivity.
\newblock In {\em {SODA}}, pages 323--342. {SIAM}, 2019.

\bibitem{alon2003}
Noga Alon, Baruch Awerbuch, and Yossi Azar.
\newblock The online set cover problem.
\newblock In {\em Proceedings of the thirty-fifth annual ACM symposium on
  Theory of computing}, pages 100--105, 2003.

\bibitem{BampisDEdLMS21}
Evripidis Bampis, Christoph D{\"{u}}rr, Thomas Erlebach, Murilo~Santos de~Lima,
  Nicole Megow, and Jens Schl{\"o}ter.
\newblock Orienting (hyper)graphs under explorable stochastic uncertainty.
\newblock In {\em {ESA}}, volume 204 of {\em LIPIcs}, pages 10:1--10:18, 2021.

\bibitem{Behnezhad2022}
Soheil Behnezhad, Avrim Blum, and Mahsa Derakhshan.
\newblock Stochastic vertex cover with few queries.
\newblock In {\em {SODA}}, pages 1808--1846. {SIAM}, 2022.

\bibitem{Blum2020}
Avrim Blum, John~P. Dickerson, Nika Haghtalab, Ariel~D. Procaccia, Tuomas
  Sandholm, and Ankit Sharma.
\newblock Ignorance is almost bliss: Near-optimal stochastic matching with few
  queries.
\newblock {\em Oper. Res.}, 68(1):16--34, 2020.

\bibitem{bruce05uncertainty}
R.~Bruce, M.~Hoffmann, D.~Krizanc, and R.~Raman.
\newblock Efficient update strategies for geometric computing with uncertainty.
\newblock {\em Theory of Computing Systems}, 38(4):411--423, 2005.

\bibitem{ChaplickHLT21}
Steven Chaplick, Magn{\'{u}}s~M. Halld{\'{o}}rsson, Murilo~S. de~Lima, and
  Tigran Tonoyan.
\newblock Query minimization under stochastic uncertainty.
\newblock {\em Theor. Comput. Sci.}, 895:75--95, 2021.

\bibitem{chvatal1979}
Vasek Chvatal.
\newblock A greedy heuristic for the set-covering problem.
\newblock {\em {Mathematics of Operations Research}}, 4(3):233--235, 1979.

\bibitem{DeshpandeHK16}
Amol Deshpande, Lisa Hellerstein, and Devorah Kletenik.
\newblock Approximation algorithms for stochastic submodular set cover with
  applications to boolean function evaluation and min-knapsack.
\newblock {\em {ACM} Trans. Algorithms}, 12(3):42:1--42:28, 2016.

\bibitem{Dinur2014}
Irit Dinur and David Steurer.
\newblock Analytical approach to parallel repetition.
\newblock In {\em {STOC}}, pages 624--633. {ACM}, 2014.

\bibitem{dobson1982}
Gregory Dobson.
\newblock Worst-case analysis of greedy heuristics for integer programming with
  nonnegative data.
\newblock {\em Mathematics of Operations Research}, 7(4):515--531, 1982.

\bibitem{erlebach16cheapestset}
T.~Erlebach, M.~Hoffmann, and F.~Kammer.
\newblock Query-competitive algorithms for cheapest set problems under
  uncertainty.
\newblock {\em Theoretical Computer Science}, 613:51--64, 2016.

\bibitem{Erlebach22Learning}
Thomas Erlebach, Murilo~Santos de~Lima, Nicole Megow, and Jens Schl{\"{o}}ter.
\newblock Learning-augmented query policies for minimum spanning tree with
  uncertainty.
\newblock In {\em {ESA}}, volume 244 of {\em LIPIcs}, pages 49:1--49:18.
  Schloss Dagstuhl - Leibniz-Zentrum f{\"{u}}r Informatik, 2022.

\bibitem{erlebach14mstverification}
Thomas Erlebach and Michael Hoffmann.
\newblock Minimum spanning tree verification under uncertainty.
\newblock In {\em {WG}}, volume 8747 of {\em Lecture Notes in Computer
  Science}, pages 164--175. Springer, 2014.

\bibitem{Erlebach2016}
Thomas Erlebach, Michael Hoffmann, and Frank Kammer.
\newblock Query-competitive algorithms for cheapest set problems under
  uncertainty.
\newblock {\em Theoretical Computer Science}, 613:51--64, 2016.

\bibitem{feder03medianqueries}
T.~Feder, R.~Motwani, R.~Panigrahy, C.~Olston, and J.~Widom.
\newblock Computing the median with uncertainty.
\newblock {\em SIAM Journal on Computing}, 32(2):538--547, 2003.

\bibitem{Ghuge2021}
Rohan Ghuge, Anupam Gupta, and Viswanath Nagarajan.
\newblock The power of adaptivity for stochastic submodular cover.
\newblock In {\em {ICML}}, volume 139 of {\em Proceedings of Machine Learning
  Research}, pages 3702--3712. {PMLR}, 2021.

\bibitem{GoemansV06}
Michel~X. Goemans and Jan Vondr{\'{a}}k.
\newblock Covering minimum spanning trees of random subgraphs.
\newblock {\em Random Struct. Algorithms}, 29(3):257--276, 2006.

\bibitem{Goemans2006}
Michel~X. Goemans and Jan Vondr{\'{a}}k.
\newblock Stochastic covering and adaptivity.
\newblock In {\em {LATIN}}, volume 3887 of {\em Lecture Notes in Computer
  Science}, pages 532--543. Springer, 2006.

\bibitem{GolovinK11}
Daniel Golovin and Andreas Krause.
\newblock Adaptive submodularity: Theory and applications in active learning
  and stochastic optimization.
\newblock {\em J. Artif. Intell. Res.}, 42:427--486, 2011.

\bibitem{Grandoni2013}
Fabrizio Grandoni, Anupam Gupta, Stefano Leonardi, Pauli Miettinen, Piotr
  Sankowski, and Mohit Singh.
\newblock Set covering with our eyes closed.
\newblock {\em {SIAM} J. Comput.}, 42(3):808--830, 2013.

\bibitem{Gupta2023}
Anupam Gupta, Gregory Kehne, and Roie Levin.
\newblock Set covering with our eyes wide shut.
\newblock {\em CoRR}, abs/2304.02063, 2023.

\bibitem{halldorsson19sortingqueries}
M.~M. HalldÃ³rsson and M.~S. de~Lima.
\newblock Query-competitive sorting with uncertainty.
\newblock In {\em MFCS}, volume 138 of {\em LIPIcs}, pages 7:1--7:15, 2019.

\bibitem{erlebach08steiner_uncertainty}
Michael Hoffmann, Thomas Erlebach, Danny Krizanc, Mat{\'{u}}s Mihal{\'{a}}k,
  and Rajeev Raman.
\newblock Computing minimum spanning trees with uncertainty.
\newblock In {\em {STACS}}, volume~1 of {\em LIPIcs}, pages 277--288. Schloss
  Dagstuhl - Leibniz-Zentrum f{\"{u}}r Informatik, Germany, 2008.

\bibitem{ImNZ16}
Sungjin Im, Viswanath Nagarajan, and Ruben van~der Zwaan.
\newblock Minimum latency submodular cover.
\newblock {\em {ACM} Trans. Algorithms}, 13(1):13:1--13:28, 2016.

\bibitem{kahan91queries}
Simon Kahan.
\newblock A model for data in motion.
\newblock In {\em {STOC}}, pages 267--277. {ACM}, 1991.

\bibitem{KambadurNN17}
Prabhanjan Kambadur, Viswanath Nagarajan, and Fatemeh Navidi.
\newblock Adaptive submodular ranking.
\newblock In {\em {IPCO}}, volume 10328 of {\em Lecture Notes in Computer
  Science}, pages 317--329. Springer, 2017.

\bibitem{Kolliopoulos2002}
Stavros~G. Kolliopoulos and Neal~E. Young.
\newblock Tight approximation results for general covering integer programs.
\newblock In {\em {FOCS}}, pages 522--528. {IEEE} Computer Society, 2001.

\bibitem{Kolliopoulos2005}
Stavros~G. Kolliopoulos and Neal~E. Young.
\newblock Approximation algorithms for covering/packing integer programs.
\newblock {\em J. Comput. Syst. Sci.}, 71(4):495--505, 2005.

\bibitem{MaeharaY20}
Takanori Maehara and Yutaro Yamaguchi.
\newblock Stochastic packing integer programs with few queries.
\newblock {\em Math. Program.}, 182(1):141--174, 2020.

\bibitem{maehara2020}
Takanori Maehara and Yutaro Yamaguchi.
\newblock Stochastic packing integer programs with few queries.
\newblock {\em Mathematical Programming}, 182(1):141--174, 2020.

\bibitem{megow17mst}
N.~Megow, J.~Mei{\ss}ner, and M.~Skutella.
\newblock Randomization helps computing a minimum spanning tree under
  uncertainty.
\newblock {\em SIAM Journal on Computing}, 46(4):1217--1240, 2017.

\bibitem{meissner18querythesis}
J.~Mei{\ss}ner.
\newblock {\em Uncertainty Exploration: Algorithms, Competitive Analysis, and
  Computational Experiments}.
\newblock PhD thesis, Technischen Universit\"{a}t Berlin, 2018.

\bibitem{NavidiKN20}
Fatemeh Navidi, Prabhanjan Kambadur, and Viswanath Nagarajan.
\newblock Adaptive submodular ranking and routing.
\newblock {\em Oper. Res.}, 68(3):856--877, 2020.

\bibitem{rajagopalan1998}
Sridhar Rajagopalan and Vijay~V Vazirani.
\newblock Primal-dual rnc approximation algorithms for set cover and covering
  integer programs.
\newblock {\em SIAM Journal on Computing}, 28(2):525--540, 1998.

\bibitem{ShmoysS04}
David~B. Shmoys and Chaitanya Swamy.
\newblock Stochastic optimization is (almost) as easy as deterministic
  optimization.
\newblock In {\em {FOCS}}, pages 228--237. {IEEE} Computer Society, 2004.

\bibitem{vazirani2001}
Vijay~V Vazirani.
\newblock {\em Approximation algorithms}, volume~1.
\newblock Springer, 2001.

\bibitem{Vondrak07}
Jan Vondr{\'{a}}k.
\newblock Shortest-path metric approximation for random subgraphs.
\newblock {\em Random Struct. Algorithms}, 30(1-2):95--104, 2007.

\bibitem{Wang22}
Weina Wang, Anupam Gupta, and Jalani Williams.
\newblock Probing to minimize.
\newblock In {\em {ITCS}}, volume 215 of {\em LIPIcs}, pages 120:1--120:23.
  Schloss Dagstuhl - Leibniz-Zentrum f{\"{u}}r Informatik, 2022.

\bibitem{Wolsey82}
Laurence~A. Wolsey.
\newblock An analysis of the greedy algorithm for the submodular set covering
  problem.
\newblock {\em Comb.}, 2(4):385--393, 1982.

\end{thebibliography}
\appendix

\section{Comparison with Maehara and Yamaguchi~\cite{maehara2020}}
\label{app:yamaguchi}
We consider the framework by Maehara and Yamaguchi~\cite{maehara2020} on the set selection problem.
Since their algorithm is designed for maximization problems, we consider the maximization variant of \minset, i.e., we have to find the set $S \in \fs$ of \emph{maximum} $w(S)$ and determine the corresponding value. We remark that all our results also translate to the maximization variant (cf.~\Cref{app:maxset}).

The algorithm by Maehara and Yamaguchi, in each iteration, solves the LP-relaxation of the \emph{optimistic} version of the given ILP, which assumes $w_i = U_i$ for all $I_i \in \mathcal{I}$.
In this case, the ILP under consideration formalizes the set selection problem (and \emph{not} the query minimization problem as formalized by~\eqref{eq:SetSelection}).
The following LP-relaxation formulates the optimistic LP for a set selection instance $(\ui,\fs)$:

\begin{equation*}
	\begin{array}{lll}
		\max &\sum_{I_i \in \ui} x_i \cdot U_i\\
		\text{s.t. }& \sum_{I_i \in S} x_i \ge y_S \cdot |S| &\forall S \in \fs\\
		& \sum_{I_i \in E} x_i \le \sum_{S \in \fs} y_S \cdot |S|\\
		& \sum_{S \in \fs} y_S = 1 &\\
		& 0 \le x_i \le 1& \forall I_i \in \ui\\
		& 0 \le y_S \le 1& \forall S \in \fs\\
	\end{array}
\end{equation*}
Here variable $y_S$ models whether set $S$ is selected as the set of maximum value ($y_S = 1$) or not ($y_S = 0$) and the third constraint makes sure that, at least integrally, exactly one set is selected.
The variables $x_i$ model whether an interval $I_i$ is part of the selected set or not, and the first two constraints ensure that (integrally) exactly the members of the selected set $S$ (with $y_S = 1$) are selected.
Note that we use this ILP instead of~\eqref{eq:minset} because the algorithm by Maehara and Yamaguchi requires variables that correspond to elements that can be queried. 
The algorithm in each iteration solves the LP-relaxation to obtain an optimal fractional solution $(x,y)$, and queries each $I_i$ with probability $x_i$.

In the following, we give an instance of the set selection problem under stochastic explorable uncertainty for which the algorithm has a competitive ratio of $\Omega(n)$.
Let $\mathcal{I} = \{I_0,\ldots,I_n\}$ with $I_0 = (d \cdot n-\epsilon,d \cdot n+\epsilon)$ and $I_i = (0, d + \epsilon)$, $i>0$, for some large $d$ and some small $\epsilon > 0$.
Let $\fs = \{S_1,S_2\}$ with $S_1 = \{I_0\}$ and $S_2 = \{I_1,\ldots,I_n\}$. 
Assume uniform distributions and consider the algorithm that starts by querying $I_0$ and afterwards queries the intervals of $S_2$ in an arbitrary order. 
In expectation, this algorithm only needs a constant number of queries to solve the instance.

The algorithm by Maehara and Yamaguchi on the other hand, in the first iteration, solves the LP-relaxation and obtains the optimal solution $(x,y)$ with $x_0 = 0$ and $x_i = 1$ for all $i > 1$.
This means that all elements of $S_2$ are queried with a probability of $1$.
Thus, the algorithm queries at least $n$ elements, which implies a competitive ratio of $\Omega(n)$. This means that applying the algorithm by Maehara and Yamaguchi~\cite{maehara2020} to the set selection problem does not improve upon the adversarial lower bound.

\section{Equivalence of \minset and Solving~\eqref{eq:SetSelection}}
\label{app:equivalence}

In this paper, we heavily exploit that solving \minset is equivalent to solving~\eqref{eq:SetSelection}. With the following lemma, we show that this is indeed the case.

\begin{restatable}{lemma}{IPEquivalence}
	Solving \minset is equivalent to solving~\eqref{eq:SetSelection}.
\end{restatable}

\begin{proof}
	We show the lemma by proving the following claim: 	A query set $Q \subseteq \ui$ is feasible for \minset if and only if vector $x$, with $x_i = 1$ for all $I_i \in Q$ and $x_i = 0$ otherwise, is a feasible solution for the corresponding~\eqref{eq:SetSelection}.
	
	Let $Q$ be feasible for \minset, and let $x$ be a vector with $x_i = 1$ for all $I_i \in Q$ and $x_i = 0$ otherwise. 
	By definition, each feasible solution for \minset must query all non-trivial elements of some cheapest set $S^*$ with $w(S^*) = w^*$ as this is the only way to determine the value $w^*$.
	Let $N(S^*)$ denote those non-trivial elements, then we can rewrite the initial lower limit of $S^*$ as $L_{S^*} = w^* + \sum_{I_i \in N(S^*)} (L_i - w_i)$ (the lower limits of trivial elements are covered by $w^*$). 
	This implies $w^* - L_{S^*} = \sum_{I_i \in N(S^*)} (w_i - L_i)$. As $N(S^*) \subseteq Q$, we get $\sum_{I_i \in S^*} x_i \cdot (w_i -L_i) = \sum_{I_i \in S^* \cap Q} (w_i -L_i) \ge w^* - L_{S^*}$. Thus, $x$ satisfies the constraint for $S^*$.
	Let $S \not= S^*$. 
	Since $Q$ is feasible, $L_{S}(Q)$ has to be at least $w^*$ as otherwise querying $Q$ would not prove that $w^*$ is the minimum set value.
	After querying $Q$, 
	the lower limit of set $S$ is 
	\begin{align*}
		L_{S}(Q) &= \sum_{I_i \in S \setminus Q} L_i + \sum_{I_i \in S \cap Q} w_i\\
		&= \sum_{I_i \in S} L_i - \sum_{I_i \in S \cap Q} L_i + \sum_{I_i \in S \cap Q} w_i\\
		&= L_{S} + \sum_{I_i \in S \cap Q} (w_i - L_i).
	\end{align*}
	Thus, $L_{S}(Q) =  L_{S} + \sum_{I_i \in S \cap Q} (w_i - L_i) \ge w^*$ must hold, which implies $\sum_{I_i \in S \cap Q} (w_i - L_i) \ge w^* - L_{S}$ and $\sum_{I_i \in S} x_i \cdot (w_i - L_i) \ge w^* - L_{S}$.
	We can conclude that $x$ is feasible.
	
	For the other direction consider a feasible solution $x$ for~\eqref{eq:SetSelection} and the corresponding set $Q = \{I_i \mid x_i = 1\}$.
	Consider some cheapest set $S^*$.
	As $\sum_{I_i \in N(S^*)} (w_i-L_i) = w^* - L_{S^*}$ and $\sum_{I_i \in S^* \setminus N(S^*)} (w_i-L_i) = \sum_{I_i \in S^* \setminus N(S^*)} (w_i - w_i) = 0$, for $x$ to be feasible it must hold $x_i = 1$ for all $I_i \in N(S^*)$. Thus, $Q$ contains all non-trivial elements of some cheapest set $S^*$.
	To show that $Q$ is a feasible solution for \minset, it remains to show that $L_{S}(Q) \ge w^*$ for all $S \not= S^*$.
	Consider an arbitrary $S \not= S^*$.
	As $\sum_{I_i \in S} x_i \cdot (w_i-L_i) \ge (w^* - L_{S})$, we have $\sum_{I_i \in S \cap Q} (w_i - L_i) \ge w^* - L_{S}$, which implies $L_{S}(Q) = L_{S} + \sum_{I_i \in S \cap Q} (w_i - L_i) \ge w^*$.
\end{proof}

\section{Hardness of Approximation}
\label{app:inapprox}

In the following, we show that~\eqref{eq:SetSelection} is not only a special case of the multiset multicover problem but also contains the hard instances of this problem.

Erlebach et al.~\cite{Erlebach2016} showed that offline \minset, for the problem variant where it is not necessary to compute the value $w^*$, is NP-hard via reduction from vertex cover. 
In the reduction by~\cite{Erlebach2016} all intervals of the set $S^*$ with $w(S^*) = w^*$ are trivial and, thus,
the result translates to the problem variant where one has to compute $w^*$.
In the following, we strengthen this result by showing that the offline problem is as hard to approximate as~\setcover.

In~\setcover, we are given a set of elements $U=\{1,\ldots,n\}$ and a family of sets $\bar{\fs}=\{\bar{S}_1,\ldots,\bar{S}_m\}$ with $\bar{S}_j \subseteq U$.
The goal is to find a subset $H \subseteq \bar{\fs}$ of minimum cardinality such that $\bigcup_{\bar{S}_j \in H} \bar{S}_j = U$. 

\begin{theorem}
	\label{thm:verification_hardness_approx}
	There is an approximation-factor preserving reduction from \setcover to offline \minset.
\end{theorem}

\begin{proof}
	Given an instance $(U,\bar{\fs})$ of \setcover, we construct an offline \minset instance as follows:
	\begin{enumerate}
		\item Add a trivial interval $I_r=\{w_r\}$.
		\item Add a single set $C=\{I_r\}$.
		\item For each $j \in U$, add a set $S_j$.
		\item For each $\bar{S}_i \in \bar{\fs}$:
		\begin{enumerate}
			\item Add an interval $I_i=(L_i,U_i)$ with $L_i = 0$, $U_i = w_r + \delta$ and $w_i = w_r + \epsilon$ for a common $\delta > \epsilon > 0$ and some infinitesimally small $\epsilon > 0$.
			\item For each $j \in \bar{S}_i$, add interval $I_i$ to set $S_j$.
		\end{enumerate}
	\end{enumerate}
	
	This reduction clearly runs in polynomial time. 
	To finish the proof, we show the following claim:
	\emph{There is a \setcover solution $H$ of cardinality $k$ if and only if there is a feasible query set $Q$ for the constructed offline \minset instance with  $|Q|=k$.}

	By definition of the constructed instance, set $C = \{I_r\}$ is the set of minimum value $w^* = w_r$.
	Each feasible query set $Q$ for the offline \minset instance must prove that $L_{S}(Q) \ge w_r$ holds for each $S \in \fs\setminus\{C\}$. Recall that $L_{S}(Q)$ is the lower limit of $S$ after querying $Q$.	
	By definition of the constructed intervals and sets, a query set $Q$ is feasible if and only if $|Q \cap S| \ge 1$ for each $S \in \fs\setminus\{C\}$, i.e., $Q$ has to contain at least one element of each $S \in \fs\setminus\{C\}$.
	
	For the first direction, consider an arbitrary set cover $H$ for the given \setcover instance and construct $Q = \{ I_i \mid \bar{S}_i \in H\}$. 
	Clearly $|Q| = |H|$.
	Since $H$ is a set cover, each $j \in U$ is contained in at least one $\bar{S}_i \in H$.
	If $j \in U$ is contained in $\bar{S}_i$, then, by construction, $I_i$ is contained in $S_j$.
	Thus, as $H$ covers all elements $j \in U$, set $Q$ contains at least one member of each $S_j \in \fs\setminus\{C\}$ and, therefore, is a feasible query set.
	
	For the second direction, consider an arbitrary feasible query set $Q$ of the constructed instance and construct $H = \{\bar{S}_i \mid I_i\in Q\}$.
	Clearly, $|Q| = |H|$.
	Since $Q$ is feasible, it contains at least one member of each $S_j \in \fs \setminus \{C\}$.
	If $I_i \in S_j$ is contained in $Q$, then, by construction, set $\bar{S}_i \in H$ covers element $j$.
	As $Q$ contains at least one member of each $S_j \in \fs \setminus \{C\}$, it follows that $H$ covers $U$.
\end{proof}

Dinur and Steurer~\cite{Dinur2014} showed that it is NP-hard to approximate \setcover within a factor of $(1-\alpha) \cdot \ln n$ for any $\alpha>0$, where $n$ is the number of elements in the instance, via a reduction running in time $n^{1/\alpha}$.
Consider the construction of
\Cref{thm:verification_hardness_approx}.
Since the sets in the constructed offline \minset instance correspond to the elements in the input \setcover instance, the construction implies the following corollary.

\begin{coro}
	\label{coro:hardness}
	For every $\alpha > 0$, it is NP-hard to approximate offline \minset within a factor of $(1-\alpha) \cdot \ln m$, where $m=|\fs|$ is the number of sets. The reduction runs in time $m^{1/ \alpha}$.
\end{coro}

We show that \Cref{thm:verification_hardness_approx} and \Cref{coro:hardness} apply also to \minset under stochastic explorable uncertainty, even if the precise value $w_i$ of each $I_i$ is drawn independently and uniformly at random from $(L_i,U_i)$.

\inapprox*

\begin{proof}
	The hardness of approximation for offline \minset follows directly from \Cref{coro:hardness}.
	
	We continue to show the statement on \minset under uncertainty with uniform distributions.
	Consider the reduction of \Cref{thm:verification_hardness_approx} with $w_r$ towards $0$ and/or $\delta$ towards $\infty$.
	With $w_r$ running towards $0$ and/or $\delta$ running towards $\infty$, the probability that it is sufficient to query one $I_i \in S_j$ to show that $w_r \le L_{S_j}(Q)$, for some query set $Q$, goes towards $1$.
	Thus, the probability that any set $Q$ that contains at least one member of each $S \in \fs$ is feasible goes towards one as well.
	Thus, $\lim_{w_r \rightarrow 0} \EX[\OPT] = \lim_{\delta \rightarrow \infty} \EX[\OPT] = |H^*|$, where $H^*$ is the optimal solution for the input \setcover instance.
	Therefore, by~\Cref{thm:verification_hardness_approx},
	in order to be $((1-\alpha) \cdot \ln m)$-competitive, the query strategy has to compute an $((1-\alpha) \cdot \ln n)$-approximation for set cover.
	This implies NP-hardness.
\end{proof}

\section{The Maximization Variant of \minset}
\label{app:maxset}

Consider the set selection problem \maxset, which has the same input as \minset, but now the goal is to determine a set of \emph{maximum value} and to determine the corresponding value.

In \maxset, a set $Q \subseteq \ui$ is \emph{feasible} if all non-trivial elements of a set $S^*$ of maximum value $w^*$ are in $Q$ and $U_{S}(Q) \le w^*$ for all sets $S \not= S^*$.
If $Q$ would not contain all non-trivial elements of some set $S^*$ of maximum value, then the maximum value $w^*$ would still be unknown and the problem would not be solved yet.
If $U_{S}(Q) \ge w^*$ for some $S$, then both $S^*$ and $S$ could still be of maximum value or not and we have not yet determined the set of maximum value.
Thus, the problem would not be solved yet.
Our goal is again to adaptively query a feasible query set and minimize the number of queries.

In the following, we briefly sketch why all our results on \minset transfer to \maxset.
Analogous to \minset, one can show that the following ILP characterizes the problem.
That is, a query set $Q$ is feasible if and only if vector $x$ with $x_i = 1$ if $I_i \in Q$ and $x_i = 0$ otherwise is feasible for the ILP.
\begin{equation}\tag{{\sc MaxSetIP}}\label{eq:maxset}
	\begin{array}{lll}
		\min &\sum_{I_i \in \ui} x_i\\
		\text{s.t. }& \sum_{I_i \in S} x_i \cdot (U_i - w_i) \ge (U_{S} - w^*) &\forall S \in \fs\\
		& x_i \in \{0,1\}& \forall I_i \in \ui
	\end{array}
\end{equation}

In the offline setting, the ILP is the exact same special case of the multiset multicover problem as~\eqref{eq:SetSelection}, which suffices to observe that all our observations on offline \minset translate to offline \maxset.
Under uncertainty, the only difference to \minset is, that, in contrast to~\eqref{eq:SetSelection}, a small value $w_i$ leads to a larger coefficient $(U_i-w_i)$ and a small value $w^*$ leads to larger right-hand sides.
Using the inverse balancing parameter $\bar{\tau} = \min_{I_i \in \ui} \bar{\tau}_i$ with $\bar{\tau}_i = \PR[w_i \le \frac{U_i + L_i}{2}]$, it is not hard to see, that, even under uncertainty, \minset and \maxset are essentially the same problem.
Thus, all our results translate.

\end{document}